\documentclass[10pt]{article}
\usepackage[margin=1in]{geometry}

\usepackage{amstext,amsfonts,amsthm,amsmath,amssymb}

\usepackage[utf8]{inputenc}
\usepackage{microtype}
\usepackage{graphicx,tikz}
\usepackage{comment}
\usepackage{url}
\usepackage{xspace}
\usepackage{todonotes}
\usepackage{enumitem}
\usepackage[absolute]{textpos}
\usepackage{wrapfig}

\def\cqedsymbol{\ifmmode$\lrcorner$\else{\unskip\nobreak\hfil
\penalty50\hskip1em\null\nobreak\hfil$\lrcorner$
\parfillskip=0pt\finalhyphendemerits=0\endgraf}\fi} 

\newcommand{\cqed}{\renewcommand{\qed}{\cqedsymbol}}

\newcommand{\Pot}{\Phi}
\newcommand{\adh}{\sigma}
\newcommand{\Oh}{\mathcal{O}}

\newcommand{\wh}[1]{\widehat{#1}}

\renewcommand{\leq}{\leqslant}
\renewcommand{\le}{\leqslant}
\renewcommand{\geq}{\geqslant}
\renewcommand{\ge}{\geqslant}

\newcommand{\jmaj}{\mathfrak{j}}

\newcommand{\commonthanks}{%
This research is a part of projects that have received funding from the European Research Council (ERC)
under the European Union's Horizon 2020 research and innovation programme
Grant Agreements no.~677651 (M. Cygan and Micha\l{} Pilipczuk) and 714704 (Marcin Pilipczuk).
Daniel Lokshtanov is supported by the European Research Council (ERC) under the European Union’s Horizon 2020 research and innovation programme (grant no. 715744), and the United States–Israel Binational Science Foundation grant no. 2018302.
Saket Saurabh has received funding from the European Research Council
(ERC) under the European Union’s Horizon 2020 research
and innovation programme (grant agreement No 819416), and
Swarnajayanti Fellowship (No DST/SJF/MSA01/2017-18).}
\newcommand{\pawelthanks}{Polish National Science Centre grant UMO-2015/19/N/ST6/03015}
\newcommand{\marseille}{The foundations of this paper have been developed at International Workshop on Graph Decompositions, held at CIRM Marseille in January 2015.}

\newtheorem{lemma}{Lemma}[section]

\newtheorem{theorem}[lemma]{Theorem}
\newtheorem{claim}[lemma]{Claim}
\theoremstyle{definition}
\newtheorem{definition}[lemma]{Definition}

\title{Randomized contractions meet lean decompositions\thanks{%
\commonthanks{}
The research of P. Komosa is supported by \pawelthanks{}.
\marseille{}}}

\author{ 
  Marek Cygan\thanks{
    Institute of Informatics, University of Warsaw, Poland, \texttt{cygan@mimuw.edu.pl}.
  }
  \and
  Pawe\l{} Komosa\thanks{
    Institute of Informatics, University of Warsaw, Poland, \texttt{p.komosa@mimuw.edu.pl}.
  }
  \and
  Daniel Lokshtanov\thanks{
    University of California, Santa Barbara, USA, \texttt{daniello@ucsb.edu}.
  }
  \and
  Marcin Pilipczuk\thanks{
    Institute of Informatics, University of Warsaw, Poland, \texttt{marcin.pilipczuk@mimuw.edu.pl}.
  }
  \and 
  Micha\l{} Pilipczuk\thanks{
    Institute of Informatics, University of Warsaw, Poland, \texttt{michal.pilipczuk@mimuw.edu.pl}.
  }
  \and
  Saket Saurabh\thanks{
    Department of Informatics, University of Bergen, Norway and Institute of Mathematical Sciences, India, \texttt{saket@imsc.res.in}.
  }
  \and
  Magnus Wahlstr\"{o}m\thanks{
    Royal Holloway, University of London, United Kingdom, \texttt{Magnus.Wahlstrom@rhul.ac.uk}.
  }
}

\date{}

\begin{document}

\maketitle

\begin{textblock}{20}(0, 12.0)
\includegraphics[width=40px]{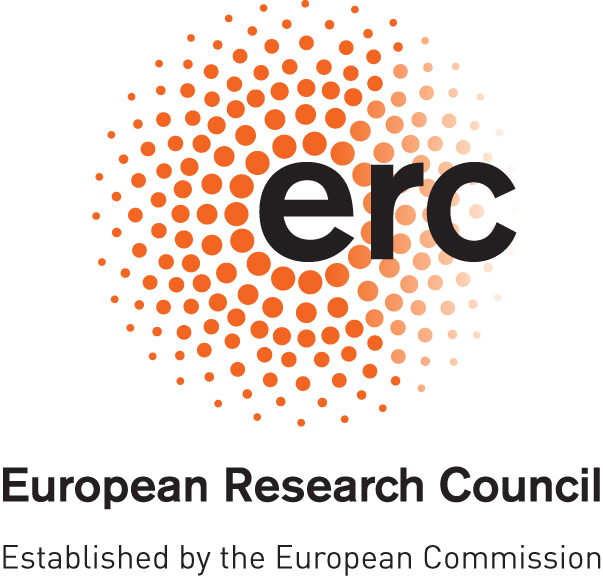}%
\end{textblock}
\begin{textblock}{20}(-0.25, 12.4)
\includegraphics[width=60px]{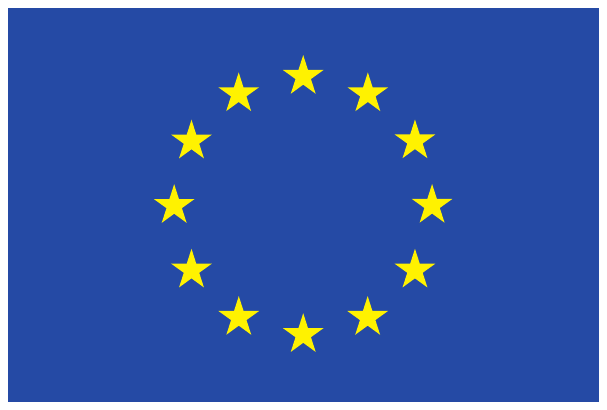}%
\end{textblock}

\begin{abstract}
We show an algorithm that, given an $n$-vertex graph $G$ and a parameter $k$, 
in time $2^{\Oh(k \log k)} n^{\Oh(1)}$ finds a tree decomposition of $G$ with the
following properties:
\begin{itemize}
\item every adhesion of the tree decomposition is of size at most $k$, and
\item every bag of the tree decomposition is $(i,i)$-unbreakable in $G$ for every $1 \leq i \leq k$.
\end{itemize}
Here, a set $X \subseteq V(G)$ is $(a,b)$-unbreakable in $G$ if for every separation $(A,B)$
of order at most $b$ in $G$, we have $|A \cap X| \leq a$ or $|B \cap X| \leq a$. 

The resulting tree decomposition has arguably best possible adhesion size bounds
and unbreakability guarantees. Furthermore, the parametric factor in the running time bound
is significantly smaller than in previous similar constructions. 
These improvements allow us to present parameterized algorithms for
\textsc{Minimum Bisection}, \textsc{Steiner Cut}, and \textsc{Steiner Multicut}
with improved parameteric factor in the running time bound. 

The main technical insight is to adapt the notion of \emph{lean decompositions} of
Thomas and the subsequent 
construction algorithm of Bellenbaum and Diestel to the parameterized setting.




\end{abstract}

\section{Introduction}

Since the pioneering work of Marx~\cite{Marx06}, 
the study of graph separation problems has been a large and lively subarea of parameterized complexity. 
It led to the development of many interesting algorithmic techniques, including
important separators and shadow removal~\cite{ChenLLOR08,dir-mwc,LokshtanovM13,MarxR14,RazgonO09},
branching algorithms based on half-integral relaxations~\cite{mwc-a-lp,sylvain,Iwata17,magnus,IwataYY17},
matroid-based algorithms for preprocessing~\cite{ms2,ms1}, the treewidth reduction technique~\cite{MarxOR13},
and, what is the most relevant for this work, the framework of randomized contractions~\cite{randcontr,KT11}.

The work of Marx~\cite{Marx06} left a number of fundamental questions open, including the parameterized complexity
of the \textsc{$p$-Way Cut} problem: given a graph $G$ and integers $p$ and $k$, can one delete at most $k$
edges from $G$ to obtain a graph with at least $p$ connected components? 
We remark that it is easy to reduce the problem to the case when $G$ is connected and $p \leq k+1$.
Marx proved $\mathsf{W}[1]$-hardness of the vertex-deletion variant of the problem, but the question about fixed-parameter tractability of the edge-deletion
variant remained open until Kawarabayashi and Thorup settled it in affirmative in 2011~\cite{KT11}.

In their algorithm, Kawarabayashi and Thorup introduced a useful recursive scheme, based on so-called \emph{$(q,k)$-good edge separations}.
For a graph $G$, an \emph{edge cut} is a pair $A, B \subseteq V(G)$ such that $A \cup B = V(G)$ and $A \cap B = \emptyset$. 
The {\em{order}} of an edge cut $(A,B)$ is $|E(A,B)|$.
An edge cut $(A,B)$ is a \emph{$(q,k)$-good edge separation} if the order of $(A,B)$ is at most $k$,
both $A$ and $B$ are connected, and $|A|,|B| > q$. 
In short, Kawarabayashi and Thorup showed that an intricate recursive step is applicable in the \textsc{$p$-Way Cut} problem if one finds a $(q,k)$-good edge separation 
for some $q = 2^{\Oh(k \log k)}$. 
%
Hence, it remains to show how to find such an edge cut $(A,B)$ and how the algorithm should work in the absence of such a cut.
The algorithm of Kawarayashi and Thorup employs involved arguments originating in the graph minor theory both to deal with this case
and to find the desired edge cut $(A,B)$ for recursion. These arguments, unfortunately, imply a large overhead in the 
running time bound and are problem-specific.

A year later, Chitnis et al.~\cite{randcontr-FOCS,randcontr} replaced the arguments based on the graph minor theory with steps
based on \emph{color coding}: a simple yet powerful algorithmic technique introduced by Alon, Yuster, and Zwick in 1995~\cite{AlonYZ95}.
This approach is both arguably simpler and leads to better running time bounds. 
Furthermore, the general methodology of~\cite{randcontr} --- dubbed \emph{randomized contractions} --- turns out to be more generic
and can be applied to solve such problems as \textsc{Unique Label Cover}, \textsc{Multiway Cut-Uncut}~\cite{randcontr},
or \textsc{Steiner Multicut}~\cite{BringmannHML16}.
All the abovementioned algorithms have running time bounds of the order of $2^{\mathrm{poly}(k)} \mathrm{poly}(n)$
with both notions of $\mathrm{poly}$ hiding quadratic or cubic polynomials.
Later, Lokshtanov et al.~\cite{LokshtanovR0Z18} showed how the idea of randomized contractions can be applied to give a reduction 
for the {\sc{CMSO}} model-checking problem from general graphs to highly connected graphs.

While powerful, the randomized contractions technique seemed to be one step short of providing 
a parameterized algorithm for the \textsc{Minimum Bisection} problem, which was an open problem at that time. 
In this problem, given a graph $G$ and an integer $k$, one asks for an edge cut $(C,D)$ of order at most $k$
such that $|C|=|D|$. 
The only step that fails is the recursive step of Kawarabayashi and Thorup itself.
A subset of the current authors in 2014~\cite{minbisection-STOC} circumvented this issue 
by replacing the recursive strategy with a dynamic programming algorithm on an appropriately constructed tree decomposition.

To properly describe this contribution, we need some more definitions on (vertex) separations.
A pair $(A,B)$ of vertex subsets in a graph $G$ is a {\em{separation}} 
if $A\cup B=V(G)$ and there is no edge with one endpoint in $A\setminus B$ and the other in $B\setminus A$; the {\em{order}} of the separation $(A,B)$ is $|A\cap B|$.

\begin{definition}[unbreakability]
Let $G$ be a graph.
A vertex subset $X \subseteq V(G)$ is \emph{$(q,k)$-unbreakable} if every separation $(A,B)$ of order at most $k$
satisfies $|A \cap X| \leq q$ or $|B \cap X| \leq q$.
A vertex subset $Y \subseteq V(G)$ is \emph{$(q,k)$-edge-unbreakable} if every edge cut $(A,B)$ of order at most $k$
satisfies $|A \cap Y| \leq q$ or $|B \cap Y| \leq q$.
\end{definition}
\noindent Observe that every set that is $(q,k)$-unbreakable is also $(q,k)$-edge-unbreakable.

The definition of a $(q,k)$-unbreakable set is stronger than the one used in~\cite{minbisection-STOC} (and in a previous
  version of the present paper). In that definition, which we may call \emph{weakly $(q,k)$-unbreakable},
  it is only required that $|(A \setminus B) \cap X| \leq q$ or $|(B \setminus A) \cap X| \leq q$.
  \footnote{To illustrate the difference, consider the case where $X=V(G)$. Then testing whether $X$ is 
  $(k,k)$-unbreakable reduces to testing whether $G$ is $(k+1)$-connected, whereas testing 
  whether $X$ is weakly $(k,k)$-unbreakable is as hard as the \textsc{Hall Set} problem, and 
  thus W[1]-hard~\cite{GaspersKOSS12}. The present version is also more suitable for Theorem~\ref{thm:decomp}.}

Let $(T,\beta)$ be a tree decomposition\footnote{In our notation, $(T,\beta)$ is a tree decomposition with $T$ being a rooted tree and $\beta(t)$ the bag at node $t \in V(T)$; full definition can be found in Section~\ref{s:prelim}.} of $G$. We define an \emph{adhesion} of $e=tt' \in E(T)$ to be the set $\adh(e) = \beta(t) \cap \beta(t')$, and an \emph{adhesion} of $t \in V(T)$ to be $\adh(t) = \adh(\{t,\text{parent}(t)\})$ or $\adh(t) = \emptyset$ if the parent of $t$ does not exist.

The main technical contribution of~\cite{minbisection-STOC} is an algorithm that, given a graph $G$
and an integer $k$, computes a tree decomposition $(T,\beta)$ of $G$ with the following properties: (i) the size of every adhesion $\adh(e)$ is bounded by a function of $k$, for $e \in E(T)$, (ii) every bag $\beta(t)$ of the decomposition is weakly $(q,k)$-unbreakable in $G$, for $t \in V(T)$.

In~\cite{minbisection-STOC}, the construction relied on advanced usage of algorithmic toolbox for graph separation problems, including the framework of important separators
and, in essence, also shadow removal.\footnote{We refer to~\cite[Chapter 8]{PAbook} for introduction to this toolbox.} This led to bounds of the form $2^{\Oh(k)}$ on the obtained value of $q$ and on the sizes of adhesions.
The construction algorithm had a running time bound of $2^{\Oh(k^2)} n^2 m$.

\paragraph*{Our results}
The main contribution of this paper is an improved construction algorithm of a decomposition with the aforementioned properties.
\begin{theorem}\label{thm:decomp}\label{thm:decomposition}
Given an $n$-vertex graph $G$ and an integer $k$, one can in time $2^{\Oh(k \log k)} n^{\Oh(1)}$ compute a rooted compact\footnote{Let $\alpha(t) = \bigcup_{s:\textrm{\ descendant\ of\ }t} \beta(s) \setminus \adh(t)$. A tree decomposition $T$ is \emph{compact} if for every non-root node $t \in T$: $G[\alpha(t)]$ is connected and $N(\alpha(t)) = \adh(t)$.} tree decomposition $(T,\beta)$ of $G$
such that 
\begin{enumerate}
\item every adhesion of $(T,\beta)$ is of size at most $k$;
\item every bag of $(T,\beta)$ is $(i,i)$-unbreakable in $G$ for every $1 \leq i \leq k$.
\end{enumerate}
\end{theorem}
\noindent Note that since every bag of the output decomposition $(T,\beta)$ is $(k,k)$-unbreakable, it is also $(k,k)$-edge-unbreakable.

The size of a bag in our decomposition is not bounded (as in most standard cases), but the high connectivity properties of it are enough for our goals. Also, the theorem implies that for every integer $k$ the decomposition always exists.



The main highlights of Theorem~\ref{thm:decomp} is the improved dependency on $k$ in the running time bound
and arguably the best possible bounds both for the unbreakability and for the adhesion sizes.
Indeed, since in the definition of $(q,k)$-unbreakability
we count the number of vertices in the separator
(we consider $|A \cap X|$ and $|B \cap X|$, as opposed to $|(A \setminus B) \cap X|$
 or $|(B \setminus A) \cap X|$), the definition degenerates for $q < k$. 
Also, if one considers problems that ask for a cut or separation of size at most $k$
(e.g., $p$-\textsc{Way Cut} or \textsc{Minimum Bisection}, discussed later), 
a tree decomposition that is disallowed to take into account separations of order $k$
(by requiring the adhesions to be strictly smaller than $k$) seems of limited use.

We also note that in fact the decomposition of~\cite{minbisection-STOC} ensures slightly more than it is in its base statement, i.e. that every bag $\beta(t)$ is suitably unbreakable in the graph $G_t$ (i.e. induced by the union of bags of descendants of $t$,  but with edges within $\adh(t)$ erased), while the decomposition of Theorem~\ref{thm:decomp} only ensures unbreakability in the whole graph $G$. This feature requires a somewhat careful treatment in the proofs of correctness of the dynamic programming algorithms on a tree decomposition. 

The polynomial factor in the running time bound of Theorem~\ref{thm:decomp} is far from a linear
one and we thus refrain from analyzing it in detail. Due to its iterative nature and 
nontrivial flow/cut arguments, already the randomized contractions technique introduced a far from
linear polynomial factor in the running time bound; in this aspect, our approach is not better. 
We see the strength of our contribution in getting the exponential factor down to $2^{\Oh(k \log k)}$ 
and obtaining best possible bounds both for the unbreakability and for the adhesion sizes.
Obtaining similar properties with an algorithm with a near-linear dependency on $n$ in the running
time bound remains an interesting and challenging open problem.

We also remark that no lower bounds for the running time of the algorithm as in Theorem~\ref{thm:decomp} are known.

\medskip

The improvements of Theorem~\ref{thm:decomp} have direct consequences for algorithmic applications:
they allow us to develop $2^{\Oh(k \log k)} n^{\Oh(1)}$-time parameterized algorithms for a number
of problems that ask for an edge cut of order at most $k$, with the most prominent one being \textsc{Minimum Bisection}.
That is, all applications mentioned below consider edge deletion problems and in their proofs we rely only on $(k,k)$-edge-unbreakability of the bags.

\begin{theorem}\label{thm:bisection}
\textsc{Minimum Bisection} can be solved in time 
$2^{\Oh(k \log k)} n^{\Oh(1)}$.
\end{theorem}

This improves the parametric factor of the running time from $2^{\Oh(k^3)}$, provided in~\cite{minbisection-STOC}, to $2^{\Oh(k\log k)}$.

In our second application, the \textsc{Steiner Cut} problem, we are given an undirected graph $G$, a set $T \subseteq V(G)$ of terminals, and integers $k,p$. 
The goal is to delete at most $k$ edges from $G$ so that $G$ has at least $p$ connected components containing
at least one terminal. This problem generalizes \textsc{$p$-Way Cut} that corresponds to the case $T=V(G)$.
\begin{theorem}\label{thm:steinercut}
\textsc{Steiner Cut} can be solved in time $2^{\Oh(k \log k)} n^{\Oh(1)}$.
\end{theorem}

This improves the parametric factor of the running time from $2^{\Oh(k^2\log k)}$, provided in~\cite{randcontr}, to $2^{\Oh(k\log k)}$.

In the \textsc{Steiner Multicut}
problem we are given an undirected graph $G$, $t$ sets of terminals $T_1,T_2,\ldots,T_t$, each of size at most $p$,
   and an integer $k$. The goal is to delete at most $k$ edges from $G$ so that every terminal set $T_i$ is separated:
   for every $1 \leq i \leq t$, there does not exist a single connected component of the resulting graph that contains the
   entire set $T_i$. 
Note that for $p=2$, the problem becomes the classic \textsc{Edge Multicut} problem.
Bringmann et al.~\cite{BringmannHML16} showed an FPT algorithm for \textsc{Steiner Multicut} when parameterized by $k+t$.
We use our decomposition theorem to improve the exponential part of the running time of this algorithm.
\begin{theorem}\label{thm:multicut}
\textsc{Steiner Multicut} can be solved in time $2^{\Oh((t + k) \log (k + t))} n^{\Oh(1)}$.
\end{theorem}

This improves the parametric factor of the running time from $2^{\Oh(k^2t\log k)}$, provided in~\cite{BringmannHML16}, to $2^{\Oh((t+k)\log (t+k))}$.

In the case of \textsc{Unique Label Cover}, \textsc{Multiway Cut-Uncut}, our decomposition can be also applied, but we were not able to achieve improvements over the dependency in the parameter(s)
  of the time complexity
of the known algorithms, i.e. $|\Sigma|^{2k}$ provided by Iwata et al. in~\cite{iwata2016half} and $2^{\Oh(k^2\log k)}$ from the pioneering randomized contraction paper~\cite{randcontr}, respectively.

\paragraph*{Our techniques}
Our starting point is the definition of a lean tree decomposition of Thomas~\cite{Thomas90}; we follow the formulation of~\cite{BellenbaumD02}.
\begin{definition}
A tree decomposition $(T,\beta)$ of a graph $G$ is called \emph{lean} if for every $t_1,t_2 \in V(T)$
and all sets $Z_1 \subseteq \beta(t_1)$ and $Z_2 \subseteq \beta(t_2)$ with $|Z_1| = |Z_2|$, either
$G$ contains $|Z_1|$ vertex-disjoint $Z_1-Z_2$ paths, or there exists an edge $e \in E(T)$ on the path from $t_1$ to $t_2$
such that $|\adh(e)| < |Z_1|$.
\end{definition}

For a graph $G$ and a tree decomposition $(T,\beta)$ that is not lean, a quadruple
$(t_1,t_2,Z_1,Z_2)$ for which the above assertion is not true is called a \emph{lean witness}.
Note that it may happen that $t_1 = t_2$ or $Z_1 \cap Z_2 \neq \emptyset$.
In particular $(s,s,Z_1,Z_2)$ is called a \emph{single bag lean witness}.
The \emph{order} of a lean witness is the minimum order of a separation $(A_1,A_2)$ such that $Z_i \subseteq A_i$ for $i=1,2$.

Bellenbaum and Diestel~\cite{BellenbaumD02} defined an improvement step that, given a tree
decomposition and a lean witness, refines the decomposition so that it is in some sense
closer to being lean. 
Given a lean witness $(t_1,t_2,Z_1,Z_2)$, the refinement step takes a minimum order separation $(A_1,A_2)$ with $Z_i \subseteq A_i$ for $i=1,2$ and rearranges the tree decomposition so that $A_1 \cap A_2$ appears
as a new adhesion on some edge of the decomposition.
Bellenbaum and Diestel introduced a potential function, bounded exponentially in $n$, that is non negative and decreases at every refinement step.
Thus, one can exhaustively apply the refinement step while a lean witness exists, obtaining (after possibly an exponential number of steps) a lean decomposition.

A simple but crucial observation connecting lean decompositions with the decomposition promised by Theorem~\ref{thm:decomposition} is that if a tree decomposition admits no single bag lean witness of order at most $k$,
then every bag is $(i,i)$-unbreakable for every $1 \leq i \leq k$.
By combining this with the fact that the refinement step applied to a lean witness of order $k'$ introduces one new adhesion of size $k'$ (and does not increase the size of other adhesions), we obtain the following.

\begin{theorem}\label{thm:exist}
For every graph $G$ and integer $k$, there exists a tree decomposition $(T,\beta)$ of $G$ such that every adhesion of $(T,\beta)$ is of size at most $k$ 
    and every bag is $(i,i)$-unbreakable for every $1 \leq i \leq k$.
\end{theorem}
\begin{proof}[Proof sketch.]
Start with a trivial tree decomposition $(T,\beta)$ that consists of a single bag $V(G)$.
As long as there exists a single bag lean witness of order at most $k$ in $(T,\beta)$, apply the refinement step of Bellenbaum and Diestel to it.
It now remains to observe that if any bag $\beta(t)$ for some $t\in T$ was not $(i,i)$-unbreakable for some $1 \leq i \leq k$,
then the separation witnessing this would give rise to a single bag lean witness for $\beta(t)$. 

\end{proof}

The formal proof of the above theorem is a subset of the proof of Theorem~\ref{thm:decomposition} without the algorithmic part that is presented in Section~\ref{s:algorithm}.

A naive implementation of the procedure of Theorem~\ref{thm:exist} runs in time exponential in $n$, while for any application in parameterized algorithms one needs an FPT algorithm with $k$ as the parameter.
To achieve this goal, one needs to overcome two obstacles. 

First, the potential provided by Bellenbaum and Diestel gave only an exponential in $n$ bound on the number of needed refinement steps.
Fortunately, one can use the fact that we only refine using single bag witnesses of bounded size to provide a different potential, this time bounded as $2^{\Oh(k)}n^{\Oh(1)}$.

The refinement step of Bellenbaum and Diestel used by us, given a single bag lean witness
for $\beta(t)$, performs a surgery on the current tree decomposition that splits some bags
into two copies such that (a) the copies have at most $k$ vertices in common, and (b)
the bag $\beta(t)$ is really split in the sense that both resulting copies have strictly
smaller size. 
The primary part of our potential sums up excesses of the sizes of bags over $(2k+1)$. 
Since one can easily trim any tree decomposition to $\Oh(n)$ bags without changing the essential
information carried by the tree decomposition, this potential is bounded by $\Oh(n^2)$. On the 
other hand, it strictly decreases if one performs a surgery to the tree decomposition that splits
a bag larger than $(2k+1)$ into two smaller bags, sharing at most $k$ vertices. 
To guarantee progress if only a bag of size at most $(2k+1)$ is split into smaller bags,
we introduce an additional secondary potential that utilizes some deeper properties of the refinement step;
thanks for the fact that it focuses on small bags, this potential can be bounded by $2^{\Oh(k)} n$.

Second, one needs to efficiently (in FPT time) verify whether a bag is $(i,i)$-unbreakable for every $1 \leq i \leq k$ and, if not, find a corresponding single bag lean witness.
We design such an algorithm using color-coding: in time $2^{\Oh(k \log k)} n^{\Oh(1)}$ we can either 
certify that all bags of a given tree decomposition are $(i,i)$-unbreakable for every $1 \leq i \leq k$
or produce a single bag lean witness of order at most $k$.

To get an intuition of how such a color coding step works, let us recall a cut-finding procedure finding $(k,k)$-good edge separation
that is already present in~\cite{randcontr}. 
Assume we want to find an edge cut $(A,B)$ of $G$ of order at most $k$ such that both 
$A$ and $B$ are connected and of size larger than $k$. Then if one randomly color every edge
white or black independently with probability $\frac{1}{2}$ each, with probability
$2^{-\Oh(k)}$ every edge of the cut $E(A,B)$ is black while there is a white connected component
of size at least $k+1$ in $G[A]$ and a white connected component of size at least $k+1$ in $G[B]$.
On the other hand, any edge cut consisting of at most $k$ black edges between two white connected
components of size larger than $k$ gives a desired edge cut (not necessarily precisely the cut $(A,B)$). 
Our algorithm for verifying $(i,i)$-unbreakability of a bag for every $1 \leq i \leq k$ leverages
the same ideas as the argument above, but requires much more delicate and precise analysis of the 
combinatorics of the situation at hand.

These ingredients lead to constructing a decomposition with guarantees as in Theorem~\ref{thm:decomposition}.

\medskip

All our applications (Theorems~\ref{thm:bisection},~\ref{thm:steinercut}, and~\ref{thm:multicut}) use the decomposition of Theorem~\ref{thm:decomposition} and follow well-paved ways
of~\cite{randcontr,minbisection-STOC} to perform bottom-up dynamic programming. 
Let us briefly sketch this for the case of \textsc{Minimum Bisection}.

Let $(G,k)$ be a \textsc{Minimum Bisection} instance and let $(T,\beta)$ be a tree decomposition of $G$ provided by Theorem~\ref{thm:decomposition}. 
The states of our dynamic programming algorithm are the straightforward ones: for every $t \in V(T)$, every $A^\adh \subseteq \adh(t)$, and every $0 \leq n^\circ \leq |\alpha(t)|$
we compute value $M[t,A^\adh,n^\circ]$ equal to the minimum order of an edge cut $(A,B)$ in $G_t$
(i.e. induced by the union of bags of descendants of $t$, 
but with edges within $\adh(t)$ erased) such that $A \cap \adh(t) = A^\adh$ and $|A \setminus \adh(t)| = n^\circ$.
Furthermore, we are not interested in cut orders larger than $k$, and we replace them with $+\infty$.
Using the unbreakability of $\beta(t)$ we can additionally require that either $A \cap \beta(t)$ or $B \cap \beta(t)$ is of size at most $k$.

In a single step of dynamic programming, one would like to populate $M[t,\cdot,\cdot]$ using the values $M[s,\cdot,\cdot]$ for all children $s$ of $t$ in $T$.
Fix a cell $M[t,A^\adh,n^\circ]$. Intuitively, one would like to iterate over all partitions $\beta(t) = A^\beta \uplus B^\beta$ with $A^\beta \cap \adh(t) = A^\adh$
and, for fixed $(A^\beta,B^\beta)$, for every child $s$ of $t$ use the cells $M[s,\cdot,\cdot]$ to read the best way to extend $(A^\beta \cap \adh(s), B^\beta \cap \adh(s))$ to $\alpha(s)$. 
However, $\beta(t)$ can be large, so we cannot iterate over all such partitions $(A^\beta,B^\beta)$.
Here, the properties of the decomposition $(T,\beta)$ come into play: since $\beta(t)$ is $(k,k)$-edge-unbreakable, and in the end we are looking for a solution to \textsc{Minimum Bisection} of order at most $k$,
we can only focus on partitions $(A^\beta,B^\beta)$ such that $|A^\beta| \leq k$ or $|B^\beta| \leq k$. While this still does not allow us to iterate over all such partitions, we can highlight important parts of them
by color coding, similarly as it is done in the leaves of the recursion in the randomized contractions framework~\cite{randcontr}.

\section{Preliminaries}\label{s:prelim}


\paragraph*{Color coding toolbox}
Let $[n]=\{1,\ldots,n\}$ for a positive integer $n$.

Many of our proofs follow the same outline as the treatment of the high-connectivity phase of the randomized contractions technique~\cite{randcontr}.
As in~\cite{randcontr}, the color coding step is encapsulated in the following lemma:

\newcommand{\randfamily}{\mathcal{F}}
\begin{lemma}[\cite{randcontr}]\label{lem:random}
Given a set $U$ of size $n$ and integers $0 \leq a,b \leq n$, one
can in time $2^{O(\min(a,b) \log (a+b))} n \log n$
construct a family $\randfamily$ of at most $2^{O(\min(a,b) \log (a+b))} \log n$
subsets of $U$ such that the following holds:
for any sets $A,B \subseteq U$ with $A \cap B = \emptyset$, $|A|\leq a$, and $|B|\leq b$,
there exists a set $S \in \randfamily$ with $A \subseteq S$ and $B \cap S = \emptyset$.
\end{lemma}

We also need the following more general version that can be obtained from Lemma~\ref{lem:random}
by a straightforward induction on $r$.

\begin{lemma}\label{lem:random2}
Given a set $U$ of size $n$ and integers $r \geq 1$, $0 \leq a_1,a_2,\ldots,a_r \leq n$
with $s = \sum_{i=1}^r a_i$, and $c = \max_{i=1}^r a_i$,
one can in time
$2^{\Oh((s-c) \cdot \log s)} n \log^{r-1} n$
construct a family $\randfamily$ of at most $2^{\Oh((s-c) \cdot \log s)} \log^{r-1} n$
functions $f \colon U \to \{1,\ldots,r\}$ such that the following holds:
for any pairwise disjoint sets $A_1,A_2,\ldots,A_r \subseteq U$ such that $|A_i|\leq a_i$ for each
$i \in [r]$,
there exists a function $f \in \randfamily$
with $A_i \subseteq f^{-1}(i)$ for every $i \in [r]$.
\end{lemma}

\paragraph*{Tree decompositions}

A \emph{tree decomposition} of a graph $G$ is a pair $(T,\beta)$ where $T$ is a tree and $\beta$ is a mapping that assigns to every $t \in V(T)$
a set $\beta(t) \subseteq V(G)$, called a \emph{bag}, such that the following holds: (i) for every $e \in E(G)$ there exists $t \in V(T)$ with $e \subseteq \beta(t)$, and (ii)
for every $v \in V(G)$ the set $\beta^{-1}(v) := \{t \in V(T) \colon v \in \beta(t)\}$ induces a connected nonempty subgraph of $T$.

For a tree decomposition $(T,\beta)$ fix an edge $e = tt' \in E(T)$. The deletion of $e$ from $T$ splits $T$ into two trees $T_1$ and $T_2$, and naturally induces a separation $(A_1,A_2)$ in $G$
with $A_i := \bigcup_{t \in V(T_i)} \beta(t)$, which we henceforth call \emph{the separation associated with $e$}.

The set $\adh_{T,\beta}(e) := A_1 \cap A_2 = \beta(t) \cap \beta(t')$ is called the \emph{adhesion} of $e$.
We suppress the subscript if the decomposition is clear from the context.

Some of our tree decompositions are rooted, that is, the tree $T$ in a tree decomposition $(T,\beta)$
is rooted at some node $r$.
For $s,t \in V(T)$ we say that \emph{$s$ is a descendant of $t$}
or that \emph{$t$ is an ancestor of $s$} if $t$ lies on the unique path from $s$ to the root;
note that a node is both an ancestor and a descendant of itself.
For a node $t$ that is not a root of~$T$, by $\adh_{T,\beta}(t)$ we mean
the adhesion $\adh_{T,\beta}(e)$ where $e$ is the edge connecting $t$ with its parent in $T$.
We extend this notation to $\adh_{T,\beta}(r) = \emptyset$ for the root $r$.
Again, we omit the subscript if the decomposition is clear from the context.

We define the following functions for convenience:
\begin{align*}
\gamma(t) &= \bigcup_{s:\textrm{\ descendant\ of\ }t} \beta(s), &
\alpha(t) &= \gamma(t) \setminus \adh(t), &
G_t &= G[\gamma(t)] - E(G[\adh(t)]).
\end{align*}
The set $\gamma(t)$ is sometimes called the \emph{cone} at $t$ while $\alpha(t)$ is the \emph{component} at $t$. 
We say that a rooted tree decomposition $(T,\beta)$ of $G$ is \emph{compact}
if for every node $t \in V(T)$ for which $\adh(t) \neq \emptyset$ we have
that $G[\alpha(t)]$ is connected and $N_G(\alpha(t)) = \adh(t)$.

We will use the following \emph{cleanup procedure} on a tree decomposition $(T,\beta)$ of a graph $G$: as long as there exists
an edge $st \in E(T)$ with $\beta(s) \subseteq \beta(t)$, contract the edge $st$ in $T$, keeping the name $t$ and the bag $\beta(t)$
at the resulting vertex. We shall say that a node $s$ and bag $\beta(s)$ \emph{disappears} in a cleanup step.
Clearly, the final result $(T',\beta')$ of a cleanup procedure is a tree decomposition of $G$ and every adhesion of $(T',\beta')$ is also an adhesion of $(T,\beta)$. Observe that $|E(T')| \leq |V(G)|$: if one roots $T'$ at an arbitrary vertex,
 going from child to parent on every edge $T'$ we forget at least one vertex of $G$, and every vertex can be forgotten only once.

It is well-known that every rooted tree decomposition can be refined to a compact one; see e.g.~\cite[Lemma 2.8]{BojanczykP16}.
For convenience, we provide a formulation of this fact suited for our needs;
a full proof can be found in Appendix~\ref{app:compactification}.

\begin{lemma}\label{lem:compactification}
Given a graph $G$ and its tree decomposition $(T,\beta)$, one can compute in polynomial time a compact tree decomposition $(\wh{T},\wh{\beta})$ of $G$ such that every bag of $(\wh{T},\wh{\beta})$ 
is a subset of some bag of $(T,\beta)$, and every adhesion of $(\wh{T},\wh{\beta})$ is a subset of some adhesion of $(T,\beta)$.
\end{lemma}

\section{Constructing a lean decomposition: Proof of Theorem~\ref{thm:decomp}}\label{s:lean}

\subsection{Refinement step}

Bellenbaum and Diestel~\cite{BellenbaumD02} defined an improvement step that, given a tree
decomposition and a lean witness, refines the decomposition so that it is in some sense
closer to being lean. We will use the same refinement step, but only for the special single bag case. Hence, in subsequent sections we
focus on finding a single bag lean witness in a current candidate tree decomposition. 
The following observation is a direct consequence of Menger's theorem.

\begin{claim}
\label{claim:equiv}
For a tree decomposition $(T,\beta)$, a node $s \in T$, and subsets $Z_1,Z_2 \subseteq \beta(s)$ with $|Z_1|=|Z_2|$,
the following two conditions are equivalent:
\begin{itemize}
  \item $(s,s,Z_1,Z_2)$ is a single bag lean witness for $(T,\beta)$,
  \item there is a separation $(A_1,A_2)$ of $G$ and a set of vertex disjoint $Z_1-Z_2$ paths $\{P_x\}_{x \in X}$, where $X = A_1 \cap A_2$,
  such that $Z_i \subseteq A_i$ for $i\in \{1,2\}$, $|X| < |Z_1|$, and $x \in V(P_x)$ for every $x \in X$.
\end{itemize}
Moreover given a single bag lean witness $(s,s,Z_1,Z_2)$ one can find the above separation $(A_1,A_2)$ and set of paths $\{P_x\}_{x \in X}$ in polynomial time.
\end{claim}
The minimum order of a separation from the second point of Claim~\ref{claim:equiv} is called the {\em{order}}
of the single bag lean witness $(s,s,Z_1,Z_2)$.

To argue that the refinement process stops after a small number of steps, or that it stops
at all, we define a potential for a graph $G$, a tree decomposition $(T,\beta)$, and an integer $k$
as follows. 
\begin{align}
\Pot^1_{G,k}(T, \beta) & = \sum_{t \in T} \max(0, |\beta(t)|-2k-1) \\
\Pot^2_{G,k}(T, \beta) & = \sum_{t \in T} (2^{\min(|\beta(t)|, 2k+1)}-1) \\
\Pot_{G,k}(T, \beta) &= 2^{2k+1} (n+1) \Pot^1_{G,k}(T, \beta) + \Pot^2_{G,k}(T, \beta).
\end{align}
Note that if $T$ has at most $n+1$ bags, then $\Pot_{G,k}(T, \beta)$ is order-equivalent to 
the lexicographical order of $(\Pot^1_{G,k}(T, \beta), \Pot^2_{G,k}(T, \beta))$.
Also note that this potential is different from the one used by Bellenbaum and Diestel in~\cite{BellenbaumD02}:
the potential of~\cite{BellenbaumD02} can be exponential in $n$ while being oblivious to the cut
size $k$.

Given a witness, a single refinement step we use is encapsulated in the following lemma, which is essentially a repetition
of the refinement process of~\cite{BellenbaumD02} with the analysis of the new potential.
We emphasize that in this part, all considered tree decompositions are unrooted.

\begin{lemma}\label{lem:vertex-refine}
Assume we are given a graph $G$, an integer $k$, a tree decomposition $(T, \beta)$ of $G$
with every adhesion of size at most $k$, a node $s \in T$, and
a single bag lean witness $(s,s,Z_1,Z_2)$ with $|Z_1| \le k+1$.
Then one can in polynomial time compute a tree decomposition $(T',\beta')$ of $G$
with every adhesion of size at most $k$ and having at most $n+1$ bags
such that $\Pot_{G,k}(T,\beta) > \Pot_{G,k}(T',\beta')$.
\end{lemma}
\begin{figure}[tb]
\begin{center}
\includegraphics[width=.95\linewidth]{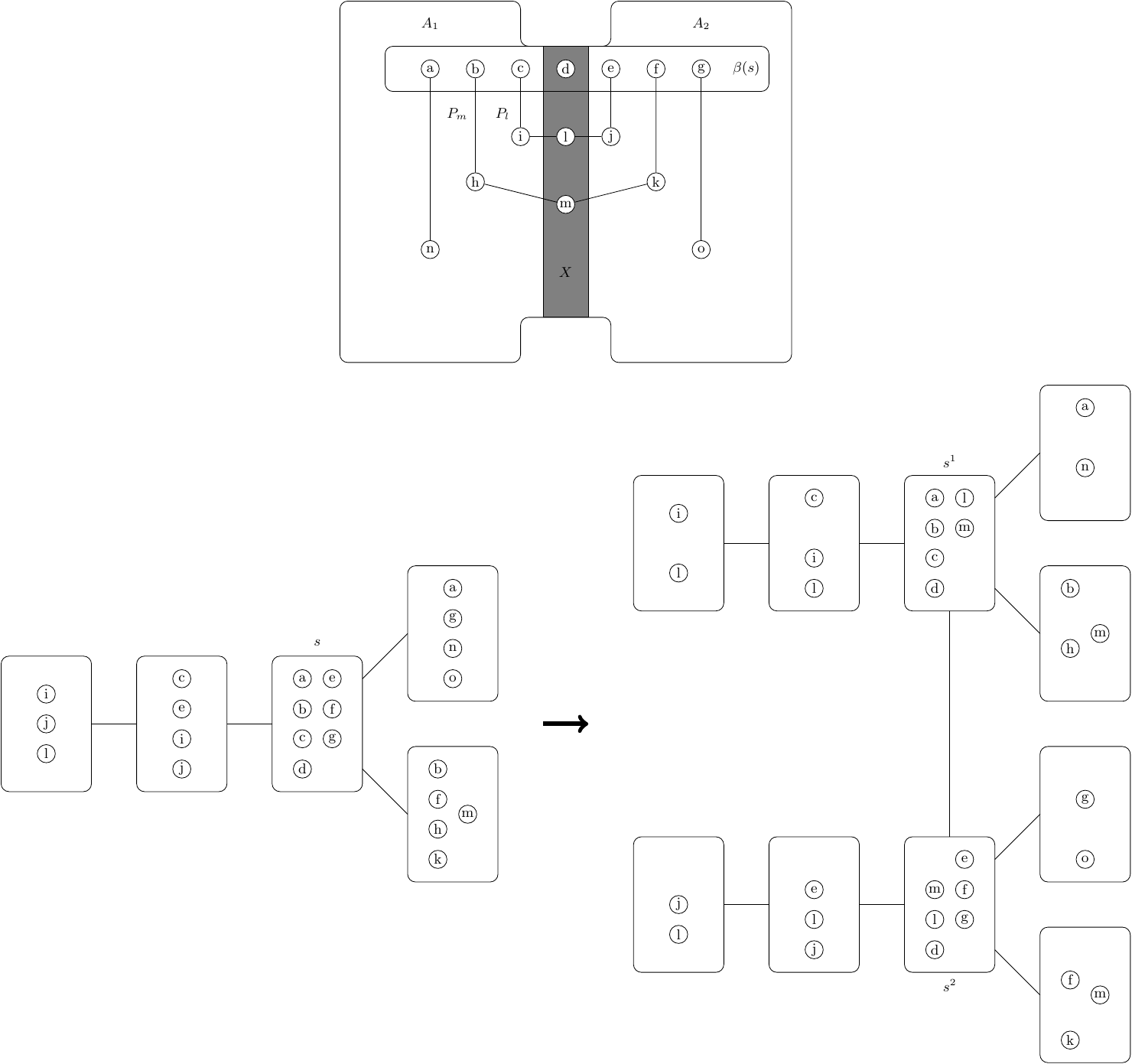}
\caption{Tree decomposition surgery in the proof of Lemma~\ref{lem:vertex-refine}.}\label{fig:s31}
\end{center}
\end{figure}

\begin{proof}
Apply Claim~\ref{claim:equiv}, yielding a separation $(A_1,A_2)$ and a family $\{P_x\}_{x \in X}$ of vertex disjoint $Z_1-Z_2$ paths, where $X=A_1\cap A_2$. 
Note that $k+1 \ge |Z_1| > |X|$, hence the order of $(A_1,A_2)$ is at most $k$.
For every $x \in X$, the path $P_x$ is split by $x$ into two subpaths $P_x^i$, $i=1,2$, each leading from $Z_i$ to $x$.
Note that $V(P_x^i) \setminus \{x\} \subseteq A_i \setminus A_{3-i}$ for $i=1,2$.

We construct a tree decomposition $(T',\beta')$ as follows.
First for every $i = 1,2$, we construct a decomposition $(T^i,\beta^i)$ of $G[A_i]$:
we start with $T^i$ being a copy of $T$, where a node $t \in V(T)$ corresponds
to a node $t^i \in V(T^i)$, and we set $\beta^i(t^i) := \beta(t) \cap A_i$ for every $t \in V(T)$.
Then for every $x \in X \setminus \beta(s)$ we take the node $t_x \in V(T)$ such that $x \in \beta(t_x)$ and $t_x$
is the closest to $s$ in $T$ among such nodes, and insert $x$ into every bag $\beta^i(t^i)$ for $t^i$
lying on the path between $s^i$ (inclusive) and $t_x^i$ (exclusive) in $T^i$. See Figure~\ref{fig:s31}.

Clearly, $(T^i,\beta^i)$ is a tree decomposition of $G[A_i]$ and $X \subseteq \beta^i(s^i)$.
We construct $(T^\circ,\beta^\circ)$ by taking $T^\circ$ to be the disjoint union of $T^1$ and $T^2$, with
the copies of the node $s$ connected by an edge $s^1s^2$, and $\beta^\circ := \beta^1 \cup \beta^2$.
Since $(A_1,A_2)$ is a separation and $X = A_1 \cap A_2$ is present in both bags $\beta^1(s^1)$, $\beta^2(s^2)$, we infer
that $(T^\circ,\beta^\circ)$ is a tree decomposition of $G$.

Finally, we apply the cleanup procedure to $(T^\circ,\beta^\circ)$, thus obtaining the final tree decomposition $(T',\beta')$ where $T'$ has at most $n$ edges.
From the properties of the separation $(A_1,A_2)$ we infer that $\beta^i(s^i) \subsetneq \beta(s)$ and
$\beta^i(s^i) \not\subseteq A_{3-i}$ for $i=1,2$. In particular, the edge $s^1s^2$ is not contracted during the cleanup
and in $(T',\beta')$ the adhesion of the edge $s^1s^2$ is exactly $X$, which is of size at most $k$.

Consider now a node $t^i$ in the tree decomposition $(T^i,\beta^i)$. The set $\beta^i(t^i) \setminus \beta(t)$
consists of some vertices of $X$, namely those vertices $x \in X \setminus \beta(s)$ for which $t$
lies on the path between $s$ (inclusive) and $t_x$ (exclusive).
On the other hand, there is a path $P_x^{3-i}$ from $x$ to $Z_2$. 
From the connectivity property (ii) of a tree decomposition this path has to cross $\beta(t)$,
as $x \not \in \beta(t)$ from the definition of $t_x$, $\beta(t)$ has to contain at least one vertex of $V(P_x^{3-i})\setminus\{x\} \subseteq A_{3-i} \setminus A_i$.
This vertex is not present in $A_i$, hence it is also not present in $\beta^i(t^i)$.
We conclude that $|\beta^i(t^i)\setminus \beta(t)|\leq |\beta(t)\setminus \beta^i(t^i)|$, implying $|\beta^i(t^i)| \leq |\beta(t)|$.
The same argumentation can be applied to every edge $e^i \in E(T^i)$ and adhesion of this edge, showing that $|\sigma^i(e^i)|\leq |\sigma(e)|$.

We infer that every adhesion of $(T^\circ,\beta^\circ)$ (and thus also of $(T',\beta')$) is of size at most $k$.
We are left with analysing the potential decrease.  We make an auxiliary claim.

\begin{claim}
  The first part of the potential is non-increasing, i.e., 
  $\Pot^1_{G,k}(T, \beta) \geq \Pot^1_{G,k}(T', \beta')$.
  Furthermore, if $|\beta(t)| > \max(|\beta^1(t^1)|, |\beta^2(t^2)|)$ for any $t \in V(T)$
  with $|\beta(t)|>2k+1$, then $\Pot_{G,k}^1(T, \beta) > \Pot_{G,k}^1(T', \beta')$.
\end{claim}
\begin{proof}
  It suffices to show that this holds for $(T^\circ,\beta^\circ)$, as the cleanup operation does not increase any of the parts of the potential.
  Fix $t \in V(T)$. We analyse the difference between the contribution to the potential
  of $t$ in $(T,\beta)$ and the copies of $t$ in $(T^\circ,\beta^\circ$).
  First, we have already argued that $|\beta^i(t^i)| \leq |\beta(t)|$
  for $i=1,2$. Consequently, if $|\beta(t)| \leq 2k+1$, then $|\beta^i(t^i)| \leq 2k+1$ for $i=1,2$
  and the discussed contributions are all equal to~$0$.
  Furthermore, if $|\beta^i(t^i)| \leq 2k+1$ for some $i=1,2$, then $\max(|\beta^{i}(t^{i})|-2k-1,0)=0$ and
  $\max(|\beta(t)|-2k-1,0) \geq \max(|\beta^{3-i}(t^{3-i})|-2k-1,0)$ due to
  $|\beta^{3-i}(t^{3-i})| \leq |\beta(t)|$, yielding the claim.
  
  Otherwise, assume that $|\beta(t)| > 2k+1$ and $|\beta^i(t^i)| > 2k+1$ for $i=1,2$.
  Note that $\beta^1(t^1) \cup \beta^2(t^2) \subseteq \beta(t) \cup X$
  while $\beta^1(t^1) \cap \beta^2(t^2) \subseteq X$ and $|X| \leq k$.
  Consequently, 
  \begin{align}
    |\beta(t)|-2k-1 &\geq |\beta(t) \setminus A_2| + |\beta(t) \setminus A_1| - 2k-1\nonumber \\
                    &\geq |(\beta(t) \setminus A_2) \cup X| + |(\beta(t) \setminus A_1) \cup X| - 4k-1\label{eq:raccoon} \\
                    &> |\beta^1(t^1)| -2k-1 + |\beta^2(t^2)| -2k-1.\nonumber
  \end{align}
  We infer that for every bag $t \in V(T)$, the contribution of $t$ to the potential
  $\Pot^1_{G,k}(T,\beta)$ is not smaller than the contribution of the two copies of $t$
  in $\Pot^1_{G,k}(T^\circ,\beta^\circ)$.  

  Finally, assume that $|\beta(t)| > \max(|\beta^1(t^1)|, |\beta^2(t^2)|)$ and $|\beta(t)| > 2k+1$. 
  If $|\beta^i(t^i)| > 2k+1$ for $i=1,2$, then by~\eqref{eq:raccoon} the potential $\Pot^1$ decreases, so assume (w.l.o.g.) that $|\beta^2(t^2)| \leq 2k+1$. 
  Then $t^2$ contributes nothing to $\Pot^1_{G,k}(T^\circ, \beta^\circ)$, while 
  the contribution of $t^1$ to $\Pot^1_{G,k}(T^\circ, \beta^\circ)$ is strictly smaller than
  the contribution of $t$ to $\Pot^1_{G,k}(T, \beta)$.   
\cqed\end{proof}

We proceed with the proof.
Note that since both $T$ and $T'$ have at most $n+1$ vertices due to the cleanup step
$\Pot^1_{G,k}(T, \beta) > \Pot^1_{G,k}(T', \beta')$ entails $\Pot_{G,k}(T, \beta) > \Pot_{G,k}(T', \beta')$. 
Hence, in the remainder we assume that for every $t \in V(T)$ with $|\beta(t)|>2k+1$
we have $\max_{i=1, 2} |\beta^i(t^i)|=|\beta(t)|$.
We argue that in this case, one of $t^1$ and $t^2$ disappears in cleanup. 
Let $t \in V(T)$ be arbitrary with $|\beta(t)|>2k+1$, and assume
w.l.o.g.\ that $|\beta(t)| = |\beta^1(t^1)|$.  
This in particular implies that $s \neq t$.

Let $X' = \beta^1(t^1) \setminus \beta(t)$; note that $X' \subseteq X$.
Recall that for every $x \in X'$ the bag $\beta(t)$ contains at least one vertex of $V(P_x^2) \setminus \{x\} \subseteq A_2 \setminus A_1$ which is no longer present in $\beta^1(t^1)$.
Since $|\beta(t)| = |\beta^1(t^1)|$, $\beta(t)$
contains  exactly one vertex $v_x$ of $V(P_x^2) \setminus \{x\}$ for every $x \in X'$
and 
$$\beta^2(t^2) = (\beta(t) \cap X) \cup \{v_x\colon x \in X'\}.$$
Let $t_\circ$ be the neighbor of $t$ that lies on the path from $t$ to $s$ in $T$.
Fix $x \in X'$; we claim that $v_x \in \beta(t_\circ)$.
This is clear if $v_x \in Z_2$ is the endpoint of $P_x^2$, because then $v_x\in \beta(s)\cap \beta(t)$ and from the connectivity property (ii) of a tree decomposition $v_x$ has to be in $\beta(t_\circ)$ as well. Otherwise, let $w_x$ be the neighbor of $v_x$ on $P_x^2$ in the direction of the endpoint in $Z_2$.
Since $v_x$ is the only vertex of $V(P_x^2) \setminus \{x\}$ in $\beta(t)$, while the endpoint of $P_x^2$ belonging to $Z_2$ lies in $\beta(s)$, 
$w_x$ belongs to a bag at a node $t'$ of $T$ that lies on the path between $t$ and $s$ in $T$.
As $w_x\notin \beta(t)$ while $v_x$ and $w_x$ are adjacent, we infer that $v_x \in \beta(t_\circ)$, as desired.
Since $v_x\in A_2$, this entails $v_x\in \beta^2(t_\circ^2)$.

The argumentation of the previous paragraph shows that $\{v_x\colon x\in X'\}\subseteq \beta^2(t_\circ^2)$.
On the other hand, from the construction we directly have $\beta(t) \cap X \subseteq \beta^2(t_\circ^2)$. 
Thus $\beta^2(t^2) \subseteq \beta^2(t_\circ^2)$ and $t^2$ disappears in the cleanup procedure.
Hence, its contribution to $\Pot_{G,k}(T', \beta')$ can be ignored.

Finally, we claim that under these conditions, $\Pot^2_{G,k}(T, \beta) > \Pot^2_{G,k}(T', \beta')$. 
Take any $t \in V(T)$. 
If either $t^1$ or $t^2$ disappears in $T'$, then the contribution of $t$ to $\Pot^2_{G,k}(T, \beta)$ is not smaller than the total contribution of $t^1$ and $t^2$ to $\Pot^2_{G,k}(T', \beta')$.
Hence, suppose that both $t^1$ and $t^2$ remain in $T'$.
Letting $\ell=|\beta(t)|$, by the assumption we have $\ell \leq 2k+1$ and the contribution
of $t$ to $\Pot^2_{G,k}(T,\beta)$ is equal to  $2^\ell-1$. On the other hand, since $\max_{i=1, 2} |\beta^i(t^i)|<|\beta(t)|$,
the total contribution of $t^1$ and $t^2$ to $\Pot^2_{G,k}(T', \beta')$ is
at most $2 \cdot (2^{\ell-1}-1) = 2^\ell - 2$, which is strictly smaller than $2^{\ell}-1$.
Since (as noted) $|\beta^i(s^i)| < |\beta(s)|$ for $i=1,2$, both $s^1$ and $s^2$ remain in $T'$, so
it follows that $\Pot_{G,k}(T, \beta) > \Pot_{G,k}(T', \beta')$.
\end{proof}

\subsection{Finding lean witness}\label{s:vertex}

In this section, we prove that if there is a bag that is not $(i,i)$-unbreakable for some $1 \leq i \leq k$, then we can efficiently find a single-bag lean witness of this bag.

Let us recall, that in the example finding a $(k,k)$-good edge separation presented in the introduction, in order to find the solution we used the fact that both $A$ and $B$ are connected and we could easily just scan the white components. In the case of a vertex separation of a single-bag lean witness, we do not have this property, but our stronger definition of the unbreakability allows us to find a connected part (Claim~\ref{cl:vertex:connected}) that
leads to the solution.

The proof follows the principles of the 
algorithms for the high-connectivity phase in the technique
of randomized contractions~\cite{randcontr}
and the \textsc{Hypergraph Painting} subroutine of~\cite{minbisection-STOC}; 
in particular, the main ingredient is color-coding. 

\begin{lemma}\label{lem:separations}
Given an $n$-vertex graph $G$, an integer $k$, and a set $S \subseteq V(G)$
with the property that every connected component $D$ of $G-S$ satisfies $|N_G(D)| \leq k$,
one can in time $2^{\Oh(k \log k)} n^{\Oh(1)}$ either
find a separation $(A,B)$ in $G$ such that the order of $(A,B)$ is some $\ell \leq k$ and it holds that $|A \cap S| > \ell$ and $|B \cap S| > \ell$, or correctly conclude
that no separation $(A,B)$ of order $\ell \leq k$ in $G$ has the property that $|A \cap S| > \ell$ and $|B \cap S| > \ell$.
\end{lemma}
\begin{proof}
Recall that we are given an $n$-vertex graph $G$, an integer $k$, and a set $S \subseteq V(G)$
such that for every connected component $D$ of $G-S$ we have $|N_G(D)| \leq k$.
We refer to the set $N_G(D)$ as the \emph{attachment} of $D$ in $S$ and we say that
$X \subseteq S$ is an \emph{attachment set} if there exists a component $D$ of $G-S$
such that $X$ is the attachment of $D$ in $S$.

Our goal is to test whether there is a separation $(A,B)$ in $G$ of order
$\ell \leq k$ such that $|A \cap S|, |B \cap S| > \ell$. 
We call such a separation $(A,B)$ a \emph{solution}.
Since the problem is easily solved in time $n^{\Oh(k)}$, we assume~$k \geq 2$.

We define $H$ to be the torso of $S$ in $G$. 
That is, we take $V(H) = S$ and two vertices $u,v \in S$
are connected by an edge in $H$ if $uv \in E(G)$ or there exists a component $D$
of $G-S$ such that $u,v\in N_G(D)$.
Note that every attachment set becomes a clique in $H$.

Assume that there exists a solution $(A,B)$ of order $\ell \leq k$ and fix one such solution.
We say that a component $D$ of $G-S$ is \emph{touched} if $D \cap A \cap B \neq \emptyset$.
A vertex $t \in S$ is \emph{touched} if either $t \in A \cap B$ 
or there exists a touched component $D$ with $t \in N_G(D)$.
Note that every component $D$ gives raise to at most $k$ touched vertices.
There are at most $k$ touched components of $G-S$ and at most $k^2$ touched vertices of $S$.

We now consider two cases. 

\paragraph*{Case 1:} We first assume that both $H[(A \setminus B) \cap S]$ and
$H[(B \setminus A) \cap S]$
contain connected components of size at least $k+1$ each.
Select any set $S_A \subseteq (A \setminus B) \cap S$ consisting of $k+1$ vertices such that $H[S_A]$ is connected,
and similarly select $S_B \subseteq (B \setminus A) \cap S$.
Let $C$ be the set of touched vertices that are not contained in $S_A \cup S_B$.

We apply Lemma~\ref{lem:random2} to the universe $S$ with $r=3$, $a_1=a_2=k+1$
and $a_3=c=k^2$. In time $2^{\Oh(k \log k)} n \log^2 n$ we obtain a family $\randfamily$
of $2^{\Oh(k \log k)} \log^2 n$ functions $f \colon S \to \{1,2,3\}$ such that 
there exists $f \in \randfamily$ with
$S_A \subseteq f^{-1}(1)$, $C \subseteq f^{-1}(2)$
and $S_B \subseteq f^{-1}(3)$.
Note that if this is the case, then $S_A$ is contained in one connected component of
$H[f^{-1}(1)]$ and $S_B$ is contained in one connected component of $H[f^{-1}(3)]$.

The following claim follows directly from the definition of touched vertices.
\begin{claim}\label{cl:vertex:use-f}
If $f \colon S \to \{1,2,3\}$ is such that 
$S_A \subseteq f^{-1}(1)$, $C \subseteq f^{-1}(2)$,
and $S_B \subseteq f^{-1}(3)$, then
the connected component of $H[f^{-1}(1)]$ containing $S_A$ is completely contained in $(A \setminus B) \cap S$
and, symmetrically, 
the connected component of $H[f^{-1}(3)]$ containing $S_B$ is completely contained in $(B \setminus A) \cap S$.
\end{claim}

By iterating over all the fewer than $|\randfamily| \cdot |S|^2 = 2^{\Oh(k \log k)} n^2 \log^2 n$
options, we guess the function $f \in \randfamily$ satisfying the assumptions of Claim~\ref{cl:vertex:use-f} 
and the connected components $A_0$ and $B_0$ of $H[f^{-1}(1)]$ and $H[f^{-1}(3)]$ that contain $S_A$ and $S_B$, respectively.
Note that $|A_0| \geq |S_A| = k+1$ and $|B_0| \geq |S_B| = k+1$. On the other hand,
since $A_0 \subseteq A \setminus B$ and $B_0 \subseteq B \setminus A$, while $|A \cap B| \leq \ell$, one can in polynomial
time find a separation $(A',B')$ of order at most $\ell$ with $A_0 \subseteq A'$ and $B_0 \subseteq B'$. While we may not necessarily obtain $(A',B') = (A,B)$, the separation
$(A',B')$ satisfies the desired properties
and can be returned by the algorithm.
This finishes the description of the first case.

\paragraph*{Case 2:} We now assume that either every connected component of $H[(A \setminus B) \cap S]$
has at most $k$ vertices, or this holds for $H[(B \setminus A) \cap S]$.
By symmetry, without loss of generality assume the former.
In particular, this implies that every connected component of $G[A \setminus B]$ contains
at most $k$ vertices of $S$.

\begin{claim}\label{cl:vertex:degH}
For every $v \in (A \setminus B) \cap S$, the degree of $v$ in $H$ is less than $k^2$
\end{claim}
\begin{proof}
Consider a vertex $w \in N_H(v) \cap B$. Since $vw \in E(H)$, either $vw \in E(G)$
or there exists a component $D$ of $G-S$ with $v,w \in N(D)$. In either case, either
$w \in A \cap B$ or $D \cap A \cap B \neq \emptyset$.
Since $|A \cap B| \leq k$ and every component $D$ of $G-S$ has attachment of size at most $k$,
there are at most $(k-1)k$ vertices of $N_H(v) \cap B$. 
Now, the connected component of $H[(A \setminus B) \cap S]$ has at most $k$ vertices, thus, $N_H(v) \cap (A \setminus B)$ has at most
$k-1$ vertices. In total, $N_H(v)$ has at most $k^2-1$ vertices.
\cqed\end{proof}

The following claim follows by trivial algebraic manipulations.
\begin{claim}\label{cl:vertex:solutionbalance}
A separation $(A',B')$ is a solution if and only if $|A' \cap B'| \leq k$
and 
\begin{equation}\label{eq:vertex:solutionbalance}
|(A' \setminus B') \cap S|> |(A' \cap B') \setminus S|\qquad\textrm{and}\qquad |(B' \setminus A') \cap S| > |(A' \cap B') \setminus S|.
\end{equation}
\end{claim}
We now establish the crucial observation: the ``imbalancedness'' requirement of~\eqref{eq:vertex:solutionbalance}
implies that a solution $(A',B')$ with inclusion-wise minimal $A' \setminus B'$ keeps $(A' \setminus B') \cap S$ connected in $H$.
\begin{claim}\label{cl:vertex:connected}
If $(A',B')$ is a solution with inclusion-wise minimal set $A' \setminus B'$ among all solutions, then
$H[(A' \setminus B') \cap S]$ is connected and, furthermore, every connected component of $G[A' \setminus B']$ contains a vertex of~$S$.
\end{claim}
\begin{proof}
First, assume that there exists a connected component $D$ of $G[A' \setminus B']$ that contains no vertex of $S$.
Then $(A' \setminus D, B' \cup D)$ is a solution as well, a contradiction to the minimality of $A' \setminus B'$. 

Second, assume $(A' \setminus B') \cap S$ can be partitioned into $K_1 \uplus K_2$ with $K_1,K_2 \neq \emptyset$
and no edge of $H$ connecting $K_1$ with $K_2$. See Figure~\ref{fig:claim37} for an illustration.
For $i=1,2$, let $C_i$ be the union of all connected components of $G[A' \setminus B']$ that contain a vertex of $K_i$.
By the construction of $H$, $C_1 \cap C_2 = \emptyset$. Since every connected component of $G[A' \setminus B']$ contains a vertex of $S$,
$C_1 \cup C_2 = A' \setminus B'$. That is, $C_1 \uplus C_2$ is a partition of $A' \setminus B'$.

For $i=1,2$, let $A_i = N_G[C_i]$ and $B_i = V(G) \setminus C_i$. Then $(A_i,B_i)$ is a separation in $G$ and:
\begin{align*}
A_i \setminus B_i & = C_i \subsetneq A' \setminus B' & \Rightarrow
(A_i \setminus B_i) \cap S &= C_i \cap S = K_i, \\
B_i \setminus A_i &= V(G) \setminus N_G[C_i] \supseteq V(G) \setminus N_G[A' \setminus B'] \supseteq B' \setminus A' & 
\Rightarrow (B_i \setminus A_i) \cap S & \supseteq (B' \setminus A') \cap S, \\
A_i \cap B_i &= N_G(C_i) \subseteq N_G(A' \setminus B') \subseteq A' \cap B'. & &
\end{align*}
Furthermore, by the construction of $H$, no component $D$ of $G-S$ contains both a vertex of $K_1$ and $K_2$ in its neighborhood. 
Consequently, $N(C_1) \setminus S$ is disjoint from $N(C_2) \setminus S$. We infer that
$$\left((A_1 \cap B_1) \setminus S\right) \cap \left((A_2 \cap B_2) \setminus S\right) = \emptyset.$$
Therefore at least one of the separations $(A_1,B_1)$ and $(A_2,B_2)$ satisfy~\eqref{eq:vertex:solutionbalance}
and thus, by Claim~\ref{cl:vertex:solutionbalance}, is a solution. That contradicts the minimality of $(A',B')$ and proves the claim.
\cqed\end{proof}

By Claim~\ref{cl:vertex:connected}, we can without loss of generality assume that $H[(A \setminus B) \cap S]$ is connected.
By the assumptions of this case, $|(A \setminus B) \cap S| \leq k$.
Let $S_A = (A \setminus B) \cap S$, $C = A \cap B \cap S$, and $S_B = N_H(S_A) \setminus C$. Observe that $S_A$, $C$, and $S_B$ are disjoint subsets of $S$.
Since $(A,B)$ is a separation in $G$, no edge of $G$ connects a vertex of $S_A$ and a vertex of $S_B$.
Furthermore, by Claim~\ref{cl:vertex:degH} we have that $|S_B| < k^3$.

We apply Lemma~\ref{lem:random2} to the universe $S$ with $r=3$, $a_1=a_2=k$
and $a_3=c=k^3$. In time $2^{\Oh(k \log k)} n \log^2 n$ we obtain a family $\randfamily$
of $2^{\Oh(k \log k)} \log^2 n$ functions $f \colon S \to \{1,2,3\}$ such that 
there exists $f \in \randfamily$ with
$S_A \subseteq f^{-1}(1)$, $C \subseteq f^{-1}(2)$
and $S_B \subseteq f^{-1}(3)$.

By the definition of $H$, for every connected component $D$ of $G-S$, if $N(D)$ contains a vertex of $S_A$, then $N(D) \cap (A \setminus B) \subseteq S_A$,
$N(D) \cap (A \cap B) \subseteq C$, and $N(D) \cap (B \setminus A) \subseteq S_B$. 
Furthermore, $S \cap N_G(S_A) \subseteq N_H(S_A) \subseteq C \cup S_B$ by the definition of $S_B$. We infer that:
\begin{claim}\label{cl:vertex:use-f-2}
If $f \colon S \to \{1,2,3\}$ is such that 
$S_A \subseteq f^{-1}(1)$, $C \subseteq f^{-1}(2)$,
and $S_B \subseteq f^{-1}(3)$, then
$S_A$ is one of the connected components of $H[f^{-1}(1)]$,
$A \cap S = N_H[S_A] \cap f^{-1}(\{1,2\})$,
and $B \cap S = S \setminus S_A$.
\end{claim}
We infer that there are at most $|S| \cdot |\randfamily| \leq 2^{\Oh(k \log k)} n \log^2 n$ choices for the set $S_A$ and, consequently,
for the sets $A \cap S$ and $B \cap S$.
Since $(A,B)$ is a solution, a minimum order separation $(A',B')$ between $A \cap S$ and $B \cap S$ is a solution as well. 
Consequently, we can find a solution if there is one as follows: for every $f \in \randfamily$ and every connected component
$S_A'$ of $H[f^{-1}(1)]$, find a minimum-order separation $(A',B')$ with $N_H[S_A'] \cap f^{-1}(\{1,2\}) \subseteq A'$
and $S \setminus S_A' \subseteq B'$ and check if $(A',B')$ is a solution. Return a solution if one is found, otherwise report that there is no solution.

This finishes the proof of Lemma~\ref{lem:separations}.
\end{proof}

\begin{figure}[bt]
\centering
\includegraphics[width=.25\linewidth]{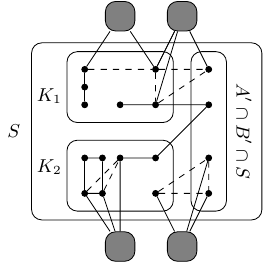}
\caption{Proof of Claim~\ref{cl:vertex:connected}. Components
  of $G-S$ are gray rectangles, edges of $E(H) \setminus E(G)$ are dashed.}
\label{fig:claim37}
\end{figure}

\subsection{The algorithm}\label{s:algorithm}

With Lemma~\ref{lem:separations} established, we now complete the proof of
Theorem~\ref{thm:decomp}.

\begin{proof}[Proof of Theorem~\ref{thm:decomp}.]
First, we can assume that $G$ is connected, as otherwise we can compute a tree decomposition
for every component separately, and then glue them up in an arbitrary fashion.

We start with a naive unrooted tree decomposition $(T,\beta)$ that has a single node whose bag is the entire vertex set. 
Then we iteratively improve it using Lemma~\ref{lem:vertex-refine},
until it satisfies the conditions of Theorem~\ref{thm:decomp}, except for compactness,
      which we will handle in the end.
We maintain the invariant that every adhesion of $(T,\beta)$ is of size at most $k$.
At every step the potential $\Pot_{G,k}(T, \beta)$ will decrease, leading to at most
$\Pot_{G,k}(T,\beta)=\Oh(4^kn^2)$ steps of the algorithm.

Let us now elaborate on a single step of the algorithm.
There is one reason why $(T,\beta)$ may not satisfy
the conditions of Theorem~\ref{thm:decomp}: namely, it contains a bag that is not $(i,i)$-unbreakable for some $1 \leq i \leq k$.

Consider a bag $S := \beta(t)$ that is not $(i,i)$-unbreakable for some $1 \leq i \leq k$.
Since every adhesion of $(T,\beta)$ is of size at most $k$, we have
that for every connected component $D$ of $G-S$ it holds that $|N_G(D)| \leq k$.
Consequently, Lemma~\ref{lem:separations} allows us to find a single-bag lean witness of order
at most $k$ for the node $t$
in time $2^{\Oh(k \log k)}\cdot n^{\Oh(1)}$. 

Suppose we uncovered a single-bag lean witness for the node $t$. 
Then we may refine the decomposition by applying Lemma~\ref{lem:vertex-refine}, and proceed iteratively with the refined decomposition.
As asserted by Lemma~\ref{lem:vertex-refine}, the potential $\Pot_{G,k}(T,\beta)$ strictly decreases in each iteration,
while the number of edges in the decomposition is always upper bounded by $n$.


Observe that the potential $\Pot_{G,k}(T,\beta)$ is always upper bounded by $2^{\Oh(k)} \cdot n^{\Oh(1)}$ 
and every iteration can be executed in time $2^{\Oh(k\log k)}\cdot n^{\Oh(1)}$.
Hence, we conclude that the refinement process finishes within the claimed time complexity and outputs an unrooted tree decomposition $(T,\beta)$
that satisfies all the requirements of Theorem~\ref{thm:decomp}, except for being compact.
This can be remedied by applying the algorithm of Lemma~\ref{lem:compactification}.
Note that neither the unbreakability of bags nor the upper bound on the sizes of adhesions can deteriorate as a result of applying the algorithm of Lemma~\ref{lem:compactification},
as every bag (resp. every adhesion) of the obtained tree decomposition is a subset of a bag (resp. an adhesion) of the original one.
\end{proof}

\section{Applications}\label{s:app}

In this section we exemplify how the decomposition provided by Theorem~\ref{thm:decomp} can be used to give faster fixed-parameter algorithms for cut problems, using the problems
{\sc{Minimum Bisection}}, {\sc{Steiner Cut}}, and {\sc{Steiner Multicut}} as examples.
Similarly as in~\cite{minbisection-STOC}, the idea is to give a bottom-up dynamic programming algorithm working on the constructed tree decomposition.
Each step of this dynamic programming boils down to solving the problem on a highly connected (formally, $(k,k)$-unbreakable) bag, 
where the tables pre-computed for children serve as ``black-boxes'' attached to adhesions between the bag and its children bags.
The key is to use color coding as in the high-connectivity phase of the randomized contractions technique~\cite{randcontr}.

A dynamic programming of this kind for {\sc{Minimum Bisection}} was already presented in~\cite{minbisection-STOC}. Unfortunately, there are multiple details that does not allow us to use this presentation as is. In particular, this dynamic programming used directly leads to the worse running time than from Theorem~\ref{thm:bisection}.
In what follows, we present the main ideas of our algorithm and
explain how it relates to the previous one. 


 In a single step of the dynamic programming, for $t \in T$, we aim to compute a size of a minimum cut $(A,B)$ of $G_t$ with the property $|A \cap \beta(t)| \leq k$ that splits $\adh(t)$ in a prescribed way and has prescribed size of $A$. The $(k,k)$-unbreakability of $\beta(t)$ allows us to make the first aforementioned assumption. 
Let $(A,B)$ be the sought minimum cut. Define $A^\ast := A \cap \beta(t)$; we have assumed that $|A^\ast| \leq k$.

We say that an adhesion of $\beta(t)$ is \emph{broken} if is not fully contained in $A$ or in $B$. Let $B^\ast$ be the set of vertices from $B$ that are either: the endpoints of the edges in the cut or the vertices in some broken adhesion. We will prove in our analysis that there are at most $k$ broken adhesions.
This in turn implies $|B^\ast| = \Oh(k^2)$, as every adhesion is of size at most $k$.

Let us color every vertex $v \in \beta(t)$ white or black independently with probability $\frac{1}{2}$ each. The set $A^\ast$ will be colored white and $B^\ast$ black with probability at least $2^{-\Oh(k^2)}$,
  as $|A^\ast| \leq k$ and $|B^\ast| = \Oh(k^2)$ and hence we request the color of $\Oh(k^2)$ vertices. 
  
Next, we define a graph $H$ to be a torso of $\beta(t)$, i.e., $H$ is constructed from the subgraph of $G$ induced by $\beta(t)$ by turning into a clique every adhesion of $\beta(t)$.
We consider white connected components of $H$.
As we colored vertices of $B^\ast$ black, those components have to be completely contained either in $A^\ast$ or completely disjoint with $A^\ast$. We inspect white connected components of $\beta(t)$ one-by-one and try to add
them to the constructed candidate for $A^\ast$. The definition of $H$ provides us
significant independence of the decisions what allows us to use multidimensional 
knapsack-type dynamic programming to find the size of the cut $(A,B)$. The randomized time $2^{\Oh(k^2)}$ can be boosted to $2^{\Oh(k \log(k))}$ and derandomized. We will present an already derandomized version using Lemma~\ref{lem:random}.

In contrast, the color coding algorithm from~\cite{minbisection-STOC} works in two levels of randomization. The first level colors edges and adhesions of a given bag $\beta(t)$ to properly color the $A$ sidde (i.e., the smaller side) of the cut, the edges of the cut and broken adhesions (to identify them later). The second level colors the candidates for broken adhesions with quantiative values corresponding to partitions of those adhesions (as the authors prove their numbers to be bounded) to properly assign them. Then, they
inspect proper components of $\beta(t)$ in dynamic programming fashion to find the solution. Our algorithm colors just the vertices of $\beta(t)$ finding the correct partitions of broken adhesions in a single phase. The method from the previous algorithm applied to our decomposition from Theorem~\ref{thm:decomp} would give $2^{\Oh(k^2 \log(k))}$ factor in the running time bound. We utilize the better bounds in our decomposition to simplify the color coding and reduce the time. On the other hand,
our dynamic programming applied to the decomposition from~\cite{minbisection-STOC}
would give doubly-exponential running time bound.

Also, as mentioned in the introduction, the decomposition of~\cite{minbisection-STOC} ensures that every bag $\beta(t)$ is suitably unbreakable in the graph $G_t$, while the decomposition of Theorem~\ref{thm:decomp} only ensures unbreakability in the whole graph $G$. This feature requires a somewhat careful treatment in the proof of correctness of the algorithm, and there are more details of similar nature regarding the fine analysis of the running time bound. 

Thus, we decided to include the full description of dynamic programming for {\sc{Minimum Bisection}}, noting that it follows the approach proposed and executed in~\cite{minbisection-STOC}. Two other problems {\sc{Steiner Cut}}, and {\sc{Steiner Multicut}} that were previously solved by randomized contraction technique in \cite{randcontr} and \cite{BringmannHML16} follow similar dynamic programming routines and we present them in the subsequent subsections.


\subsection{Minimum Bisection}\label{ss:bisection}
\newcommand{\cnstr}{\mathcal{C}}
\newcommand{\fbalance}{\mathfrak{b}}
\newcommand{\fcost}{\mathfrak{c}}
\newcommand{\symM}{M'}
\newcommand{\lasta}{\overrightarrow{a}}
\newcommand{\geti}{\overline{i}}

Our first application of the decomposition of Theorem~\ref{thm:decomp}
is an algorithm for the \textsc{Minimum Bisection} problem. Recall that in this problem, we are given a graph $G$ with even number
of vertices and an integer $k$, and the task is to find an edge cut $(A,B)$ of order
at most $k$ such that $|A| = |B|$.
We prove Theorem~\ref{thm:bisection}. That is,
we show that, with help of the decomposition of Theorem~\ref{thm:decomp},
one can obtain a fixed-parameter algorithm with better dependency on the parameter
than the one of~\cite{minbisection-STOC}.

\begin{proof}[Proof of Theorem~\ref{thm:bisection}.]
Let $(G,k)$ be an input to \textsc{Minimum Bisection} and let $n = |V(G)|$.
Without loss of generality, assume $k \geq 2$.

We start by invoking the algorithm of Theorem~\ref{thm:decomp} to $G$ and $k$,
obtaining a rooted compact
tree decomposition $(T,\beta)$ of $G$ whose every bag is $(k,k)$-edge-unbreakable
and every adhesion is of size at most $k$. The running time of this step 
is $2^{\Oh(k \log k)} n^{\Oh(1)}$.

We perform a bottom-up dynamic programming
algorithm on $(T,\beta)$. For every node $t \in V(T)$, every set $A^\adh \subseteq \adh(t)$,
and every integer $0 \leq n^\circ \leq |\alpha(t)|$, we compute an integer
$M[t, A^\adh, n^\circ] \in \{0, 1, 2, \ldots, k, +\infty\}$ with the following properties.
\begin{enumerate}[label={(\alph*)}]
\item If $M[t,A^\adh,n^\circ] \neq +\infty$, then there exists an edge cut $(A,B)$
of $G_t$ such that:\label{i:bis:cmpl}
\begin{itemize}
\item $A \cap \adh(t) = A^\adh$,
\item $|A \cap \alpha(t)| = n^\circ$,
\item $|A \cap \beta(t)| \leq k$, and
\item the order of $(A,B)$ is at most $M[t,A^\adh,n^\circ]$ in $G_t$.
\end{itemize}
\item For every edge cut $(A,B)$ of the \emph{entire graph $G$} that satisfies:\label{i:bis:sound}
\begin{itemize}
\item $A \cap \adh(t) = A^\adh$,
\item $|A \cap \alpha(t)| = n^\circ$, 
\item $|A \cap \beta(t)| \leq k$, and
\item the order of $(A,B)$ is at most $k$ in $G$,
\end{itemize}
the order of the edge cut $(A \cap \gamma(t), B \cap \gamma(t))$ is at least
$M[t,A^\adh,n^\circ]$ in $G_t$.
\end{enumerate}
Let us first formally observe that the table $M[\cdot]$ is sufficient for our purposes.
\begin{claim}\label{cl:bis:M}
$(G,k)$ is a yes-instance to \textsc{Minimum Bisection} if and only if
$M[r,\emptyset,n/2] \neq +\infty$ for the root $r$ of $T$.
\end{claim}
\begin{proof}
In one direction, note that the edge cut $(A,B)$ whose existence is asserted by Point~\ref{i:bis:cmpl} for $M[r,\emptyset,n/2] \neq +\infty$ witnesses that $(G,k)$ is a yes-instance, as $G_r = G$.

In the other direction, let $(A,B)$ be an edge cut of $G$ of order at most $k$
such that $|A| = |B| = n/2$. Since $\beta(r)$ is $(k,k)$-edge-unbreakable, we have
$|A \cap \beta(r)| \leq k$ or $|B \cap \beta(r)| \leq k$; w.l.o.g. assume the former. 
Since $\adh(r) = \emptyset$, the edge cut $(A,B)$ fits into the assumptions for Point~\ref{i:bis:sound} for the cell $M[r,\emptyset,n/2]$, finishing the proof.
\cqed\end{proof}

Thus, to prove Theorem~\ref{thm:bisection}, it suffices to show how to compute in time
$2^{\Oh(k \log k)} n^{\Oh(1)}$ entries $M[t,\cdot,\cdot]$ for a fixed node $t \in V(T)$, given
entries $M[s,\cdot,\cdot]$ for all children $s$ of $t$ in $T$. Since $|\adh(t)| \leq k$,
it suffices to focus on a computation of a single cell $M[t,A^\adh,n^\circ]$.
Let $Z$ be the set of children of $t$ in $T$.

The definition of the set $M[t, A^\adh, n^\circ]$
is somewhat asymmetric, as it requires the same side of the 
separation $(A,B)$ to be of size at most $k$ in $\beta(t)$ and to contain $A^\adh$ from
the adhesion $\adh(t)$. Given values $M[t, \cdot, \cdot]$, let us define the symmetrized
variant as follows.
$$\symM[t, A^\adh, n^\circ] = \min\left(M[t, A^\adh, n^\circ],
    M[t, \adh(t) \setminus A^\adh, |\alpha(t)|-n^\circ]\right).$$
We have the following straightforward claim:
\begin{claim}\label{cl:bis:sym}
If the values $M[t, \cdot, \cdot]$ satisfy Points~\ref{i:bis:cmpl} and~\ref{i:bis:sound},
   then the values $\symM[t, \cdot, \cdot]$ satisfy the following.
\begin{enumerate}[label={(\alph*')}]
\item If $\symM[t,A^\adh,n^\circ] \neq +\infty$, then there exists an edge cut $(A,B)$
of $G_t$ such that:\label{i:bis:cmpl:sym}
\begin{itemize}
\item $A \cap \adh(t) = A^\adh$,
\item $|A \cap \alpha(t)| = n^\circ$,
\item the order of $(A,B)$ is at most $M[t,A^\adh,n^\circ]$ in $G_t$.
\end{itemize}
\item For every edge cut $(A,B)$ of the \emph{entire graph $G$} that satisfies:\label{i:bis:sound:sym}
\begin{itemize}
\item $A \cap \adh(t) = A^\adh$,
\item $|A \cap \alpha(t)| = n^\circ$, 
\item the order of $(A,B)$ is at most $k$,
\end{itemize}
the order of the edge cut $(A \cap \gamma(t), B \cap \gamma(t))$ is at least
$\symM[t,A^\adh,n^\circ]$ in $G_t$.
\end{enumerate}
\end{claim}
\begin{proof}
For Point~\ref{i:bis:cmpl:sym}, consider two cases.
If $\symM[t,A^\adh,n^\circ] = M[t,A^\adh,n^\circ]$, then observe that the edge cut $(A,B)$
asserted by Point~\ref{i:bis:cmpl} for $M[t,A^\adh,n^\circ]$ works also here.
Otherwise, if $\symM[t,A^\adh,n^\circ] = M[t, \adh(t) \setminus A^\adh, |\alpha(t)|-n^\circ]$,
then let $(A,B)$ be the edge cut asserted by Point~\ref{i:bis:cmpl} for
$M[t, \adh(t) \setminus A^\adh, |\alpha(t)|-n^\circ]$, and observe that $(B,A)$ satisfies
all the requirements.

For Point~\ref{i:bis:sound:sym}, let $(A,B)$ be an edge cut as in the statement. 
Since $\beta(t)$ is $(k,k)$-edge-unbreakable in $G$, we have that either $|A \cap \beta(t)| \leq k$
or $|B \cap \beta(t)| \leq k$. In the first case, we have that $(A,B)$ satisfies the requirements
of Point~\ref{i:bis:sound} for $M[t,A^\adh,n^\circ]$, yielding that the order of $(A,B)$
is at least $M[t, A^\adh,n^\circ]$. 
Otherwise, we have that $(B,A)$ satisfies the requirements
of Point~\ref{i:bis:sound} for $M[t,\adh(t)\setminus A^\adh,|\alpha(t)|-n^\circ]$, yielding that the order of $(B,A)$ (which is the same as order of $(A,B)$) is at least $M[t, \adh(t) \setminus A^\adh,|\alpha(t)|-n^\circ]$. This finishes the proof of the claim.
\cqed\end{proof}

Intuitively, in a single step of a dynamic programming algorithm, we would like to focus only
on partitioning $\beta(t)$ into $A$-side and $B$-side of the partition $(A,B)$, and read
the best way to partition subgraphs $G[\alpha(s)]$ for $s \in Z$ from the tables
$M[s, \cdot, \cdot]$. 
To this end, every adhesion $\adh(s)$ for $s \in Z$ serves as a ``black-box'' that, given
a partition of $\adh(s)$ and a requested balance of the partition of $\alpha(s)$,
returns a minimum-size edge cut of $G_s$. Within the same framework, one can think of edges 
$e \in E(G_t[\beta(t)])$ as ``mini-black-boxes'' that force us to pay $1$ if we put the endpoints
of $e$ into different sets. 

This motivates the following definition of a family $\cnstr$ of \emph{constraints};
every child $s \in Z$ and every
edge $e \in E(G_t[\beta(t)])$ gives raise to a single constraint.
A constraint $\Gamma \in \cnstr$ consists of:
\begin{itemize}
\item a set $X_\Gamma \subseteq \beta(t)$ of size at most $k$;
\item a nonnegative integer $n_\Gamma$;
\item a function $M_\Gamma: 2^{X_\Gamma} \times \{0,1,\ldots,n_\Gamma\} \to \{0,1,2,\ldots,k,+\infty\}$.
\end{itemize}
For a child $s \in Z$, we define a constraint $\Gamma(s)$ as:
\begin{itemize}
\item $X_{\Gamma(s)} = \adh(s) = \beta(s) \cap \beta(t)$,
\item $n_{\Gamma(s)} = |\alpha(s)|$,
\item $M_{\Gamma(s)}(A^\adh_s, n^\circ_s) = \symM[s,A^\adh_s,n^\circ_s]$ for every $A^\adh_s \subseteq X_{\Gamma(s)}$ and $0 \leq n^\circ_s \leq n_{\Gamma(s)}$.
\end{itemize}
For an edge $e \in E(G_t[\beta(t)])$, we define a constraint $\Gamma(e)$ as:
\begin{itemize}
\item $X_{\Gamma(e)} = e$,
\item $n_{\Gamma(e)} = 0$,
\item $M_{\Gamma(e)}(\emptyset, 0) = M_{\Gamma(e)}(e, 0) = 0$ and $M_{\Gamma(e)}(\{v\}, 0) = 1$
for every $v \in e$.
\end{itemize}
A \emph{balance function} is a function $f$ that assigns to every constraint $\Gamma \in \cnstr$
an integer $f(\Gamma) \in \{0, 1, \ldots, n_\Gamma\}$. 
A \emph{feasible set} is a set $A_t \subseteq \beta(t)$ with
$A_t \cap \adh(t) = A^\adh$ and $|A_t| \leq k$. 
Given a feasible set $A_t$ and a balance function $f$, the \emph{balance} and \emph{cost} 
of $A_t$ and $f$ are defined as
\begin{align}
\fbalance(A_t,f) &= |A_t \setminus \adh(t)| + \sum_{\Gamma \in \cnstr} f(\Gamma),\nonumber\\
\fcost(A_t,f) &= \sum_{\Gamma \in \cnstr} M_\Gamma(A_t \cap X_\Gamma, f(\Gamma)).\label{eq:bis:cost}
\end{align}
We claim the following.
\begin{claim}\label{cl:bis:feas}
Assume that the values $M'[s,\cdot,\cdot]$ satisfy Points~\ref{i:bis:cmpl:sym} and~\ref{i:bis:sound:sym} for every $s \in Z$.
Then, the following assignment satisfies Points~\ref{i:bis:cmpl} and~\ref{i:bis:sound}:
\begin{equation}\label{eq:bis:Mdef}
M[t, A^\adh, n^\circ] := \min \left\{ \fcost(A^t, f)\ |\ A^t\mathrm{\ is\ feasible,\ }f\mathrm{\ is\ a\ balance\ function}\mathrm{,\ and\ }\fbalance(A^t, f) = n^\circ\right\}.
\end{equation}
In the above, $M[t, A^\adh, n^\circ] := +\infty$ if the right hand side exceeds $k$.
\end{claim}
\begin{proof}
Let $A_t^\ast$ and $f^\ast$ be a feasible set and a balance function that achieve the minimum in the right hand side of~\eqref{eq:bis:Mdef}.

Consider first Point~\ref{i:bis:cmpl}, and assume $A_t^\ast$ and $f^\ast$ exist and have
cost at most $k$. Since $\fcost(A_t^\ast,f^\ast) \neq +\infty$,
for every $s \in Z$ we have $M_{\Gamma(s)}(A_t^\ast \cap X_{\Gamma(s)}, f^\ast(\Gamma(s))) \neq +\infty$, that is, $\symM[s, A_t^\ast \cap \adh(s), f^\ast(\Gamma(s))] \neq \emptyset$. Let $(A_s,B_s)$ be the edge cut for the node $s$ whose existence is asserted by Point~\ref{i:bis:cmpl:sym} for the cell $\symM[s, A_t^\ast \cap \adh(s), f^\ast(\Gamma(s))]$. Note that $A_s \cap \adh(s) = A_t^\ast \cap \adh(s)$ by definition.

This allows us to define $A = A_t^\ast \cup \bigcup_{s \in Z} A_s$ with the properties that
$A \cap \beta(t) = A_t^\ast$ and $A \cap \gamma(s) = A_s$ for every $s \in Z$. 
In particular, we have $|A \cap \beta(t)| \leq k$ and $A \cap \adh(t) = A^\adh$.
Consequently, by the definition of the balance of $(A_t^\ast, f^\ast)$ and the properties of $(A_s,B_s)$ asserted by Point~\ref{i:bis:cmpl}, we infer that $|A \setminus \adh(t)| = \fbalance(A_t^\ast, f^\ast)$. 
Here, recall that $f^\ast(\Gamma(s)) = |A_s \setminus \adh(s)|$, that is, the balance
integer $f^\ast(\Gamma(s))$ does not count the vertices in the adhesion $\adh(s)$.

Furthermore, by the definition of the constraints $\Gamma(e)$ for $e \in E(G_t[\beta(t)])$ 
and since the order of $(A_s,B_s)$ is at most $M[s, A_t^\ast \cap \adh(s), f^\ast(\Gamma(s))]$
in $G_s$,
a direct check shows that the order of $(A,B)$ in $G_t$ is at most $\fcost(A_t^\ast, f^\ast)$.
Here, recall the definition $G_t = G[\gamma(t)]-E(G[\adh(t)])$, i.e., $G_t$ does not contain
any edge inside $\adh(t)$.

Hence, $(A,B)$ satisfies all the properties for Point~\ref{i:bis:cmpl} for the cell $M[t, A^\adh, n^\circ]$.

Let us now consider Point~\ref{i:bis:sound}, and let $(A,B)$ be an edge cut in $G$ of order at most $k$ such that $A \cap \adh(t) = A^\adh$, $|A \cap \alpha(t)| = n^\circ$, and $|A \cap \beta(t)| \leq k$.

Let $A_t = A \cap \beta(t)$ and define a balance function $f$ as follows:
$f(\Gamma(s)) = |A \cap \alpha(s)|$ for every $s \in Z$ and 
$f(\Gamma(e)) = 0$ for every $e \in E(G_t[\beta(t)])$. 
By minimality of $\fcost(A_t^\ast, f^\ast)$, we have that
$\fcost(A_t^\ast, f^\ast) \leq \fcost(A_t, f)$. 
Hence, it suffices to show that the order of $(A \cap \gamma(t), B \cap \gamma(t))$ in $G_t$ is at least
$\fcost(A_t, f)$. 

To this end, consider an edge $e \in E(G_t) \cap E(A,B)$. If $e \in E(G_t[\beta(t)])$, then 
the constraint $\Gamma(e) \in \cnstr$ contributes $1$ in the sum in~\eqref{eq:bis:cost}.
Otherwise, $e \in E(G_s)$ for a unique $s \in Z$. 
In this case, observe that $(A,B)$ satisfies the prerequisites for 
Point~\ref{i:bis:sound:sym} for the entry $\symM[s, A \cap \adh(s), |A \cap \alpha(s)|]$.
Note that this is the same entry as $\symM[s, A_t \cap \adh(s), f(\Gamma(s))]$,
which is equal to $M_{\Gamma(s)}(A_t \cap X_{\Gamma(s)}, f(\Gamma(s)))$. 
Consequently, every $e \in E(G_s) \cap E(A,B)$ is counted in the summand corresponding
to $\Gamma(s)$ in~\eqref{eq:bis:cost}.
This finishes the proof of the claim.
\cqed\end{proof}

By Claim~\ref{cl:bis:feas}, our goal is minimize $\fcost(A_t, f)$ among all
feasible sets $A_t$ and balance functions $f$ with $\fbalance(A_t,f) = n^\circ$.
Assume this minimum is finite and at most $k$, and
fix a minimizing argument $(A_t^\ast, f^\ast)$. 
We say that a constraint $\Gamma \in \cnstr$ is \emph{broken} if both
$X_\Gamma \cap A_t^\ast$ and $X_\Gamma \setminus A_t^\ast$ are nonempty. 
We claim the following.
\begin{claim}\label{cl:bis:broken}
There are at most $\fcost(A_t^\ast, f^\ast) \leq k$ broken constraints.
\end{claim}
\begin{proof}
It suffices to show that every broken constraint contributes positive value to the sum
in~\eqref{eq:bis:cost}. This is straightforward for a constraint $\Gamma(e)$ for $e \in E(G_t[\beta(t)])$.

Consider now a constraint $\Gamma(s)$ for $s \in Z$. Observe that $\Gamma(s)$ contributes
$\symM[s, A_t^\ast \cap \adh(s), f(\Gamma(s))]$ to the sum in~\eqref{eq:bis:cost}.
By Point~\ref{i:bis:cmpl:sym}, there exists an edge cut $(A,B)$ in $G_s$ of order
at most $\symM[s, A_t^\ast \cap \adh(s), f(\Gamma(s))]$ with $A \cap \adh(s) = A_t^\ast \cap \adh(s)$. However, due to the compactness of $(T,\beta)$,
and the assumption that $\Gamma(s)$ is broken
(i.e., both $A_t^\ast \cap \adh(s)$ and $\adh(s) \setminus A_t^\ast$ are nonempty),
  any such an edge cut has positive order. This finishes the proof of the claim.
\cqed\end{proof}

Let $B^\ast$ be the set of all those vertices
$v \in \beta(t) \setminus A_t^\ast$ for which there exist a broken constraint $\Gamma \in \cnstr$
with $v \in X_\Gamma$. By Claim~\ref{cl:bis:broken}, we have that $|B^\ast| \leq k^2$.
We invoke Lemma~\ref{lem:random} for the universe $\beta(t)$ and integers $k$ and $k+k^2$.
We obtain a family $\randfamily$ of size $2^{\Oh(k \log k)} \log n$ such that
there exists $S \in \randfamily$ with $A_t^\ast \subseteq S$, but $S \cap (B^\ast \cup (\adh(t) \setminus A_t^\ast)) = \emptyset$. 
Note that this in particular implies $S \cap \adh(t) = A^\adh = A_t^\ast \cap \adh(t)$.
We call such a set $S$ \emph{lucky}.

Consider now an auxiliary graph $H$ with $V(H) = \beta(t)$ and $uv \in E(H)$ if and only if
$u \neq v$ and there exists a constraint $\Gamma \in \cnstr$ with $u,v\in X_\Gamma$.
By the definition of constraints $\Gamma(e)$, we have that $G_t[\beta(t)]$ is a subgraph of
$H$, but in $H$ we also turn all adhesions $\adh(s)$ for $s \in Z$ into cliques.
Observe the following.

\begin{claim}\label{cl:bis:stains}
For a lucky set $S$,
every connected component of $H[S]$ is either completely contained in $A_t^\ast$
or completely disjoint with $A_t^\ast$.
\end{claim}
\begin{proof}
Assume the contrary: let $uv \in E(H)$ be an edge with $u, v \in S$ but
$u \in A_t^\ast$ and $v \notin A_t^\ast$. By the definition of $H$,
there exists a constraint $\Gamma \in \cnstr$ with $u,v \in X_\Gamma$.
Since $u \in A_t^\ast$ but $v \notin A_t^\ast$, the constraint $\Gamma$ is broken.
However, then $v \in B^\ast$, contradicting the fact that $S$ is lucky.
\cqed\end{proof}

Claim~\ref{cl:bis:stains} motivates the following approach. We try every set $S \in \randfamily$,
and proceed under the assumption that $S$ is lucky. 
We inspect connected components of $H[S]$ one-by-one and either try to add them to the constructed
candidate set for $A_t^\ast$ or not. Claim~\ref{cl:bis:stains} asserts that one could construct
$A_t^\ast$ in this manner. The definition of $H$ implies that for every constraint $\Gamma \in \cnstr$, the set $X_\Gamma$ intersects at most one component of $H[S]$. This gives significant
independence of the decisions, allowing us to execute a multidimensional knapsack-type dynamic
programming, as between different connected components of $H[S]$ we need only to keep intermediate values of the
balance and cost of the constructed set and balance function.

Let us proceed with formal details. 
For every $S \in \randfamily$ with $S \cap \adh(t) = A^\adh$ (which is a necessary condition for being lucky), we proceed as follows.
Let $C_1,C_2,\ldots,C_\ell$ be the connected components of $H[S]$ and for $1 \leq i \leq \ell$, 
let $\cnstr_i$ be the set of constraints $\Gamma \in \cnstr_i$ with
$X_\Gamma \cap C_i \neq \emptyset$.
By the definition of $H$,
  the sets $\cnstr_i$ are pairwise disjoint.
Let $\cnstr_0 = \cnstr \setminus \bigcup_{i=1}^\ell \cnstr_i$ be the remaining constraints, i.e., the constraints $\Gamma \in \cnstr$ with $X_\Gamma \cap S = \emptyset$.
It will be convenient for us to denote $C_0 = \emptyset$ to be a component accompanying
$\cnstr_0$.
For $I \subseteq \{0,1,2,\ldots,\ell\}$, denote $C_I = \bigcup_{i \in I} C_i$.
Furthermore, let $I^\adh \subseteq \{1,2,\ldots,\ell\}$ be the set of these indices $j$
for which $C_j \cap \adh(t) \neq \emptyset$.
For $0 \leq i \leq \ell$, denote $\cnstr_{\leq i} = \bigcup_{j \leq i} \cnstr_j$.

Let $m = |\cnstr|$. We order constraints in $\cnstr$ as $\Gamma_1,\Gamma_2,\ldots,\Gamma_m$
according to which set $\cnstr_i$ they belong to. That is, if $\Gamma_a \in \cnstr_i$,
$\Gamma_b \in \cnstr_j$ and $i < j$, then $a < b$.
For $0 \leq i \leq \ell$,
let $\lasta(i) = |\cnstr_{\leq i}|$, that is, $\Gamma_{\lasta(i)} \in \cnstr_{\leq i}$
but $\Gamma_{\lasta(i)+1}$ does not belong to $\cnstr_{\leq i}$ (if exists).
Furthermore, let $\lasta(-1) = 0$.

Let $0 \leq i \leq \ell$ and let $\lasta(i-1) \leq a \leq \lasta(i)$.
An \emph{$a$-partial balance function} $f$ is a balance function
defined on constraints $\Gamma_1,\Gamma_2,\ldots,\Gamma_a$ (i.e., $f(\Gamma_b)$ is defined and belongs 
to $\{0, 1, \ldots, n_{\Gamma_b}\}$ for every $0 \leq b \leq a$).
For a set $I \subseteq \{0,1,2,\ldots,i\}$ containing $I^\adh \cap \{0,1,2,\ldots,i\} $
and an $a$-partial balance function $f$,
we define the balance and cost of $I$ and $f$ as
\begin{align*}
\fbalance_{i,a}(I, f) &= |C_I \setminus \adh(t)| + \sum_{b=1}^a f(\Gamma_b), \\
\fcost_{i,a}(I, f) &= \sum_{b=1}^a M_{\Gamma_b}(C_I \cap X_{\Gamma_b}, f(\Gamma_b)). \\
\end{align*}
Note that in the above, $C_I \cap X_{\Gamma_b} \neq \emptyset$ if and only if $\Gamma_b \in \cnstr_j$
for some $j \in I \setminus \{0\}$.

The goal of our knapsack-type dynamic programming algorithm is to compute, for every
$-1 \leq i \leq \ell$ and $0 \leq n^\bullet \leq |\alpha(t)|$
a value $Q_i[n^\bullet]$ that equals a minimum possible cost
$\fcost_{i,\lasta(i)}(I,f)$ over all $I \subseteq \{0,1,2,\ldots,i\}$ containing
$I^\adh \cap \{0,1,2,\ldots,i\}$ and $\lasta(i)$-partial balance functions
$f$ with $\fbalance_{i,\lasta(i)}(I,f) = n^\bullet$.
The next claim shows that this suffices.
\begin{claim}\label{cl:bis:dp}
We have $Q_{\ell}[n^\circ] \geq \fcost(A_t^\ast, f^\ast)$.
Furthermore, if $\fcost(A_t^\ast, f^\ast) \leq k$ and $S$ is lucky, then 
$Q_{\ell}[n^\circ] = \fcost(A_t^\ast, f^\ast)$.
\end{claim}
\begin{proof}
For the first claim, let $(I,f)$ be the witnessing arguments
for the value $Q_{\ell}[n^\circ]$. Note that $f$ is a balance function and $m = \lasta(\ell)$.
Let $A_t = C_I$. Observe that $n^\circ = \fbalance_{\ell,m}(I, f) = \fbalance(A_t,f)$
and $\fcost_{\ell,m}(I, f) = \fcost(A_t, f)$ directly from the definition.
The first claim follows from the minimality of $(A_t^\ast, f^\ast)$.

For the second claim, by Claim~\ref{cl:bis:stains} we know that every component $C_i$
is either completely contained in $A_t^\ast$ or disjoint with it. Let $I^\ast$ be the set
of indices $i \in \{1,2,\ldots,\ell\}$ for which $C_i \subseteq A_t^\ast$; note that 
$C_{I^\ast} = A_t^\ast$ and $I^\adh \subseteq I^\ast$ as $A_t^\ast \cap \adh(t) = A^\adh$.
Consequently, $\fbalance_{\ell,m}(I^\ast, f^\ast) = \fbalance(A_t^\ast, f^\ast) = n^\circ$ and $\fcost_{\ell,m}(I^\ast, f^\ast) = \fcost(A_t^\ast,f^\ast)$.
By the minimality of $Q_{\ell}[n^\circ]$, we have $Q_{\ell}[n^\circ] \leq \fcost(A_t^\ast, f^\ast)$, as desired.
\cqed\end{proof}

For the initialization step of our dynamic programming algorithm, note that for $i=-1$
we have $\lasta(i)=0$ and the minimization for $Q_{-1}[\cdot]$
takes into account only $I = \emptyset$ and $f=\emptyset$.

Fix now $0 \leq i \leq \ell$, we are to compute values $Q_i[\cdot]$. 
To this end, we use another layer of dynamic programming.
For $\lasta(i-1) \leq a \leq \lasta(i)$ and $0 \leq n^\bullet \leq |\alpha(t)|$,
we define $Q_{i,a}^\in[n^\bullet]$ to be a minimum possible cost
$\fcost_{i,a}(I,f)$ over all $I \subseteq \{0,1,2,\ldots,i\}$ containing
$I^\adh \cap \{0,1,2,\ldots,i\}$ and the index $i$, and all $a$-partial balance functions
$f$ with $\fbalance_{i,a}(I,f) = n^\bullet$.

Let us now compute the values $Q_{i,a}^\in[\cdot]$.
Since both $Q_{i,\lasta(i-1)}^\in$ and $Q_{i-1}$ use all constraints up to $\lasta(i-1)$,
we have that for every $0 \leq n^\bullet \leq |\alpha(t)|$:
$$Q_{i,\lasta(i-1)}^\in[n^\bullet] = Q_{i-1}[n^\bullet-|C_i|].$$
Here, and in subsequent formulas, we assume that a value of a cell $Q_{i,a}^\in[\cdot]$
or $Q_i[\cdot]$ equals $+\infty$ if the argument is negative or larger than $|\alpha(t)|$.
The above formula serves as the initialization step for computing values $Q_{i,a}^\in[\cdot]$.

For a single computation step, fix $\lasta(i-1) < a \leq \lasta(i)$.
The definitions of $Q_{i,a}^\in[\cdot]$ and $Q_{i,a-1}^\in[\cdot]$ differ only in the requirement
to define $f(\Gamma_a)$. Furthermore, $X_{\Gamma_a}$ intersects only the component $C_i$
(if $i > 0$), while $i$ is required to be contained in $I$ in the definition
of both $Q_{i,a}^\in[\cdot]$ and $Q_{i,a-1}^\in[\cdot]$. Consequently,
$$Q_{i,a}^\in[n^\bullet] = \min_{0 \leq n' \leq n^\bullet} \left( Q_{i,a-1}^\in[n^\bullet - n'] + M_{\Gamma_a}(C_i \cap X_{\Gamma_a}, n') \right).$$

If $i \notin I^\adh$, we need  also a second table, defined as follows.
Let $Q_{i,a}^{\notin}[n^\bullet]$ to be a minimum possible cost
$\fcost_{i,a}(I,f)$ over all $I \subseteq \{0,1,2,\ldots,i-1\}$ containing
$I^\adh \cap \{0,1,2,\ldots,i-1\}$, and all $a$-partial balance functions
$f$ with $\fbalance_{i,a}(I,f) = n^\bullet$. Note that here $I$ is required \emph{not}
to contain $i$. By a similar analysis as before, we obtain that for every $0 \leq n^\bullet \leq |\alpha(t)|$:
$$Q_{i,\lasta(i-1)}^{\notin}[n^\bullet] = Q_{i-1}[n^\bullet],$$
and for every $\lasta(i-1) < a \leq \lasta(i)$ and every $0 \leq n^\bullet \leq |\alpha(t)|$:
$$Q_{i,a}^{\notin}[n^\bullet] = \min_{0 \leq n' \leq n^\bullet} \left( Q_{i,a-1}^{\notin}[n^\bullet - n'] + M_{\Gamma_a}(\emptyset, n') \right).$$

We now show how to compute $Q_i[\cdot]$ from the values $Q_{i,\lasta(i)}^\in[\cdot]$ and $Q_{i,\lasta(i)}^{\notin}[\cdot]$ (if $i \notin I^\adh$).
Note that in the definition of $Q_i[\cdot]$, there is no requirement on whether
$I$ contains $i$ or not (as opposed to the definitions
    of $Q_{i,a}^\in[\cdot]$ and $Q_{i,a}^{\notin}[\cdot]$), unless $i \in I^\adh$.
Thus, while computing $Q_i[n^\bullet]$, which is a minimum of
$\fcost_{i,\lasta(i)}(I,f)$ over all $I \subseteq \{0,1,2,\ldots,i\}$ containing
$I^\adh \cap \{0,1,2,\ldots,i\}$ and $\lasta(i)$-partial balance functions
$f$ with $\fbalance_{i,\lasta(i)}(I,f) = n^\bullet$,
we separately consider sets $I$ that contain $i$ and the ones that do not contain $i$.
For the first case, note that the required minimum value is present in the cell
$Q_{i,\lasta(i)}^\in[n^\bullet]$, while for the second case in the cell $Q_{i,\lasta(i)}^{\notin}[n^\bullet]$. Consequently, we have that in the case $i \notin I^\adh$ it holds that
$$Q_i[n^\bullet] = \min\left(Q_{i,\lasta(i)}^\in[n^\bullet], Q_{i,\lasta(i)}^{\notin}[n^\bullet]\right),$$
while if $i \in I^\adh$ we have only the first case:
$$Q_i[n^\bullet] = Q_{i,\lasta(i)}^\in[n^\bullet].$$

The above dynamic programming algorithm computes the values $Q_\ell[\cdot]$ in polynomial time.
By Claims~\ref{cl:bis:feas} and~\ref{cl:bis:dp}, we can take $M[t, A^\adh, n^\circ]$
to be the minimum value of $Q_\ell[n^\circ]$ encountered over all choices of $S \in \randfamily$
with $S \cap \adh(t) = A^\adh$.
Claim~\ref{cl:bis:M} shows that the above suffices to conclude the proof of Theorem~\ref{thm:bisection}.
\end{proof}

\subsection{Steiner Cut and Steiner Multicut}\label{ss:steiner}
We now apply the framework to the \textsc{Steiner Cut} and \textsc{Steiner Multicut} problems.

To prove Theorems~\ref{thm:steinercut} and~\ref{thm:multicut} we introduce an auxiliary problem,
   provide an algorithm for the auxiliary problem through our framework,
   and then show how to reduce the aforementioned two problems to the auxiliary one.
   
\subsubsection{Auxiliary Multicut} 
Given a graph $G$ and a function $f : V(G) \to [p]$, the \emph{cost} of $f$ in $G$, denoted $\fcost_G(f)$ is defined as the number of edges $e = uv \in E(G)$ with $f(u) \neq f(v)$.
We use the following intermediate problem, which we call \textsc{Auxiliary Multicut}. 
The input consists of a graph $G$, integers $k,p$, terminal sets $(T_i)_{i=1}^\tau$ (not necessarily disjoint), and a set $I \subseteq [\tau] \times [p]$.
The goal is to find a function $f \colon V(G) \to [p]$ of cost at most $k$ such that for every $(i,j) \in I$, there exists $v \in T_i$ with $f(v) = j$.

\begin{theorem}\label{thm:auxcut}
\textsc{Auxiliary Multicut} on connected graphs $G$ can be solved in time $2^{\Oh((k+|I|) \log (k+|I|))} n^{\Oh(1)}$.
\end{theorem}
\begin{proof}
Let $(G,k,p,(T_i)_{i=1}^\tau,I^\circ)$ be an input to \textsc{Auxiliary Multicut}. Assume $G$ is connected and let $n = |V(G)|$.
Without loss of generality, we can assume that every set $T_i$ is nonempty and for every $j \in [p]$ there exists some $i \in [\tau]$ with $(i,j) \in I^\circ$.
Thus, the image of the sought function $f$ (henceforth called a \emph{solution}) needs to be equal the whole $[p]$.
Consequently, the connectivity of $G$ allows us to assume $p \leq k+1$, as otherwise the input instance is a no-instance: any such function $f$ would have cost larger than $k$.

We invoke the algorithm of Theorem~\ref{thm:decomp} to $G$ and $k$,
obtaining a rooted compact
tree decomposition $(T,\beta)$ of $G$ whose every bag is $(k,k)$-edge-unbreakable
and every adhesion is of size at most $k$. The running time of this step 
is $2^{\Oh(k \log k)} n^{\Oh(1)}$.

We perform a bottom-up dynamic programming algorithm on $(T,\beta)$. 
For every node $t \in V(T)$, every subset $I \subseteq I^\circ$, and every
function $f^\adh\colon \adh(t) \to [p]$ we compute a value $M[t, I, f^\adh] \in \{0, 1, 2, \ldots, k, +\infty\}$
with the following properties.

\begin{enumerate}[label={(\alph*)}]
\item If $M[t,I, f^\adh] \neq +\infty$, then there exists a function $f : V(G_t) \to [p]$ such that:\label{i:ac:cmpl}
\begin{itemize}
\item $f|_{\adh(t)} = f^\adh$,
\item for every $(i,j) \in I$ there exists $w\in V(G_t) \cap T_i$ with $f(w) = j$,
\item the cost of $f$ in $G_t$ is at most $M[t,I,f^\adh]$.
\end{itemize}
\item For every function $f \colon V(G) \to [p]$ that satisfies:\label{i:ac:sound}
\begin{itemize}
\item $f|_{\adh(t)} = f^\adh$,
\item for every $(i,j) \in I^\circ$ there exists $w \in T_i$ with $f(w) = j$, and, 
  furthermore, if $(i,j) \in I$, then there exists such $w$ in $T_i \cap V(G_t)$,
\item the cost of $f$ in $G$ is at most $k$,
\end{itemize}
the cost of $f|_{V(G_t)}$ in $G_t$ is at least $M[t,I,f^\adh]$.
\end{enumerate}

We first observe that the table $M[\cdot]$ is sufficient for our purposes.
\begin{claim}\label{cl:ac:M}
There exists a solution if and only if
$M[r,I^\circ,\emptyset] \neq +\infty$ for the root $r$ of $T$.
\end{claim}
\begin{proof}
In one direction, if $f$ is a solution, then note that
$f$ satisfies all conditions of Point~\ref{i:ac:sound} for the cell
$M[r,I^\circ,\emptyset]$. Hence, by Point~\ref{i:ac:sound}, $M[r,I^\circ,\emptyset] < +\infty$.

In the other direction, let $M[r,I^\circ,\emptyset] < +\infty$ and 
let $f$ be the function whose existence is asserted by Point~\ref{i:ac:cmpl}.
Then for every $(i,j) \in I^\circ$ there exists $w \in T_i$ with $f(w) = j$ and its cost in $G_r = G$ is at most $k$.
Consequently, $f$ is a solution.
\cqed\end{proof}

Thus, to prove Theorem~\ref{thm:auxcut}, it suffices to show how to compute in time
$2^{\Oh((|I^\circ| + k) \log (k + |I^\circ|))} n^{\Oh(1)}$ entries $M[t,\cdot,\cdot]$ for a fixed node $t \in V(T)$, given
entries $M[s,\cdot,\cdot]$ for all children $s$ of $t$ in $T$. Since $|\adh(t)| \leq k$ and $p \leq k+1$,
it suffices to focus on a computation of a single cell $M[t,I,f^\adh]$.
Let $Z$ be the set of children of $t$ in~$T$.

In a single step of a dynamic programming algorithm, we would like to focus on finding
a function $f^\beta : \beta(t) \to [p]$ that extends $f^\adh$
and read the best way to extend $f^\beta$ to subgraphs $G[\alpha(s)]$
for $s \in Z$ from the tables $M[s,\cdot,\cdot]$.
To this end, every adhesion $\adh(s)$ for $s \in Z$ serves as a ``black-box'' that, given 
$f_s^\adh : \adh(s) \to [p]$ and a request $I_s \subseteq I$, returns the minimum possible cost in $G_s$
of an extension $f_s$ of $f_s^\adh$ that satisfies the request $I_s$ (i.e., for every $(i,j) \in I_s$ there exists $w \in T_i \cap V(G_s)$ with $f_s(w) = j$).
Within the same framework, one can think of edges $e \in E(G_t[\beta(t)])$ as ``mini-black-boxes''
that force us to pay $1$ if $f^\beta$ assigns different values to the endpoints of $e$
and vertices $w \in T_i$ as ``mini-black-boxes'' that allow us to ``score'' a pair $(i,j) \in I$
if we assign $f^\beta(w) = j$.

This motivates the following definition of a family $\cnstr$ of \emph{constraints};
every child $s \in Z$, every edge $e \in E(G_t[\beta(t)])$, and every pair $((i,j),w)$ for $(i,j) \in I$, $w \in T_i \cap \beta(t)$ gives raise to a single constraint.
A constraint $\Gamma \in \cnstr$ consists of:
\begin{itemize}
\item a set $X_\Gamma \subseteq \beta(t)$ of size at most $k$;
\item a function $M_\Gamma: 2^{I} \times [p]^{X_\Gamma} \to \{0,1,2,\ldots,k,+\infty\}$.
\end{itemize}
For a child $s \in Z$, we define a \emph{child constraint} $\Gamma(s)$ as:
\begin{itemize}
\item $X_{\Gamma(s)} = \adh(s)$,
\item $M_{\Gamma(s)}(I_s,f^\adh_s) = M[s,I_s,f^\adh_s]$ for every $I_s \subseteq I$ and $f^\adh_s : \adh(s) \to [p]$.
\end{itemize}
For an edge $uv = e \in E(G_t[\beta(t)])$, we define an \emph{edge constraint} $\Gamma(e)$ as:
\begin{itemize}
\item $X_{\Gamma(e)} = e$,
\item $M_{\Gamma(e)}(\emptyset, f^\adh_e) = 0$ 
if $f^\adh_e(u) = f^\adh_e(v)$,
 $M_{\Gamma(e)}(\emptyset, f^\adh_e) = 1$ 
if $f^\adh_e(u) \neq f^\adh_e(v)$,
and $M_{\Gamma(e)}(I_e, f^\adh_e) = +\infty$ for every $I_e \neq \emptyset$ and $f^\adh_e : e \to [p]$.
\end{itemize}
For a pair $((i,j),w)$ with $(i,j) \in I$ and $w \in \beta(t) \cap T_i$ we define a \emph{terminal constraint} $\Gamma(i,j,w)$ as:
\begin{itemize}
\item $X_{\Gamma(i,j,w)} = \{w\}$,
\item $M_{\Gamma(i,j,w)}(\emptyset, \cdot) = 0$, $M_{\Gamma(i,j,w)}(\{(i,j)\},f^\adh_{i,j,w}) = 0$ if
$f^\adh_{i,j,w}(w) = j$, and $M_{\Gamma(i,j,w)}(I_{i,j,w}, f^\adh_{i,j,w}) = +\infty$ for every other $I_{i,j,w} \subseteq I$ and $f^\adh_{i,j,w} : \{w\} \to [p]$.
\end{itemize}

A \emph{responsibility assignment} is a function $\rho$ that assigns to every constraint $\Gamma \in \cnstr$ a subset $\rho(\Gamma) \subseteq I$ such that
the values $\rho(\Gamma)$ are pairwise disjoint.
For a function $f^\beta : \beta(t) \to [p]$ we define the \emph{majority value} $\jmaj(f^\beta)$ as the minimum $j \in [p]$ among values $j$
maximizing $|(f^\beta)^{-1}(j)|$ (that is, we take the smallest among the most common values of $f^\beta$).
A function $f^\beta : \beta(t) \to [p]$ is \emph{unbreakable-consistent} if either $|\beta(t)| \leq 3k$ or at most $k$ vertices of $\beta(t)$ are assigned values different
than $\jmaj(f^\beta)$.
A \emph{feasible function} is an unbreakable-consistent function $f^\beta : \beta(t) \to [p]$ that extends $f^\adh$ on $\adh(t)$.

Given a responsibility assignment $\rho$ and a feasible function $f^\beta$, their \emph{cost} is defined as
\begin{equation}\label{eq:ac:cost}
\fcost(\rho,f^\beta) = \sum_{\Gamma \in \cnstr} M_\Gamma(\rho(\Gamma), f^\beta|_{X_\Gamma}).
\end{equation}
We claim the following.
\begin{claim}\label{cl:ac:feas}
Assume that the values $M[s,\cdot,\cdot]$ satisfy Points~\ref{i:ac:cmpl} and~\ref{i:ac:sound} for every $s \in Z$. 
Then the following assignment satisfies Points~\ref{i:ac:cmpl} and~\ref{i:ac:sound} for $t$:
\begin{equation}\label{cl:ac:Mdef}
M[t,I,f^\adh] = \min \{ \fcost(\rho, f^\beta)~|~\rho\mathrm{\ is\ a\ responsibility\ assignment,\ }f^\beta\mathrm{\ is\ a\ feasible\ function,\ }\bigcup_{\Gamma \in \cnstr} \rho(\Gamma) = I \}.
\end{equation}
In the above, $M[t, I, f^\adh] := +\infty$ if the right hand side exceeds $k$.
\end{claim}
\begin{proof}
Let $\rho^\ast$ and $f^{\beta,\ast}$ be values for which the minimum of the right hand side of~\eqref{cl:ac:Mdef} is attained. 

Consider first Point~\ref{i:ac:cmpl} and assume $\rho^\ast$ and $f^{\beta,\ast}$ exist and $\fcost(\rho^\ast, f^{\beta,\ast}) \leq k$. 
In particular, from~\eqref{eq:ac:cost} we infer that for every $\Gamma \in \cnstr$ we have $M_\Gamma(\rho^\ast(\Gamma), f^{\beta,\ast}|_{X_\Gamma}) \leq k$.
For every $s \in Z$ let $f_s$ be the function extending $f^{\beta,\ast}|_{\adh(s)}$ whose existence is promised by Point~\ref{i:ac:cmpl} for the cell
$M[s, \rho^\ast(\Gamma(s)), f^{\beta,\ast}|_{\adh(s)}] \leq k$.
We claim that
$$f := f^{\beta,\ast} \cup \bigcup_{s \in Z} f_s$$
satisfies the requirements for Point~\ref{i:ac:cmpl} for the cell $M[t, I, f^\adh]$.

First, note that $f$ is well-defined as every $f_s$ agrees with $f^{\beta,\ast}$ on $\adh(s)$.
Clearly, $f$ extends $f^\adh$ as $f^{\beta,\ast}$ extends $f^\adh$.

Second, fix $(i,j) \in I$; our goal is to show a vertex $w \in V(G_t) \cap T_i$ for which $f(w) = j$.
By~\eqref{cl:ac:Mdef}, we have $\bigcup_{\Gamma \in \cnstr} \rho^\ast(\Gamma) = I$. 
Hence, there exists $\Gamma \in \cnstr$ with $(i,j) \in \rho(\Gamma)$. 
Since $M_{\Gamma}(\rho^\ast(\Gamma), f^{\beta,\ast}|_{X_{\Gamma}}) \leq k$, $\Gamma$ is not an edge constraint $\Gamma(e)$.
If $\Gamma$ is a terminal constraint, $\Gamma = \Gamma((i,j),w)$, then from the definition of a terminal constraint we obtain that $f(w) = j$ and $w \in T_i$.
Finally, if $\Gamma$ is a child constraint, $\Gamma = \Gamma(s)$ for some $s \in Z$, then since $f_s$ is a function promised by Point~\ref{i:ac:cmpl} for the cell
$M[s, \rho^\ast(\Gamma(s)), f^{\beta,\ast}|_{\adh(s)}]$ there exists $w \in V(G_s) \cap T_i$ with $f_s(w) = j$.
This concludes the proof that for every $(i,j) \in I$ there exists $w \in V(G_t) \cap T_i$ with $f(w) = j$.

Finally, let us compute the cost of $f$. Let $e = uv \in E(G_t)$ with $f(u) \neq f(v)$. 
By the definition of the graphs $G_t$ and $G_s$, $s \in Z$, $e$ either belongs to $E(G_t[\beta(t)])$ or to exactly one of the subgraphs $G_s$, $s \in Z$.
In the first case, the constraint $\Gamma(e)$ contributes $1$ to $\fcost(\rho^\ast, f^{\beta,\ast})$. 
In the second case, for every $s \in Z$, the number of edges $e = uv \in E(G_s)$ with $f(u) \neq f(v)$ equals exactly the cost of $f_s$, which 
is not larger than $M_{\Gamma(s)}(\rho^\ast(\Gamma(s)), f^{\beta,\ast}_{\adh(s)})$. 
Point~\ref{i:ac:cmpl} for the cell $M[t, I, f^\adh]$ follows.

Let $f$ be a function as in Point~\ref{i:ac:sound} for the cell $M[t, I, f^\adh]$. 
Define responsibility assignment $\rho$ as follows. Start with $\rho(\Gamma) = \emptyset$ for every $\Gamma \in \cnstr$.
For every $(i,j) \in I$, proceed as follows. By the properties of Point~\ref{i:ac:sound}, there exists $w \in V(G_t) \cap T_i$ with $f(w) = j$.
If $w \in \beta(t)$, then we insert $(i,j)$ into $\rho(\Gamma(i, j, w))$. Otherwise, $w \in V(G_s) \setminus \adh(s)$ for some $s \in Z$; we insert then
$(i,j)$ into $\rho(\Gamma(s))$. Clearly, $\rho$ is a responsibility assignment and $\bigcup_{\Gamma \in \cnstr} \rho(\Gamma) = I$.

We define $f^\beta = f|_{\beta(t)}$ and we claim that $f^\beta$ is a feasible function. Clearly, $f^\beta$ extends $f^\adh$.
To show that $f^\beta$ is unbreakable-consistent, observe that the fact that $\beta(t)$ is $(k,k)$-edge-unbreakable in $G$ with conjunction with the assumption
that the cost of $f$ is at most $k$ implies that for every partition $[p] = J_1 \uplus J_2$ either $(f^\beta)^{-1}(J_1)$ or $(f^\beta)^{-1}(J_2)$ is of size at most $k$.
If $|(f^\beta)^{-1}(\jmaj(f^\beta))| > k$, then this implies that $|(f^\beta)^{-1}([p] \setminus \{\jmaj(f^\beta)\})| \leq k$, as desired.
Otherwise, we have $|(f^\beta)^{-1}(i)| \leq k$ for every $i \in [p]$, and unless $|\beta(t)| \leq 3k$ there exist a partition $[p] = J_1 \uplus J_2$ such that
$|(f^\beta)^{-1}(J_j)| > k$ for $j=1,2$, a contradiction. This proves that $f^\beta$ is unbreakable-consistent, and thus a feasible function.

To finish the proof of the claim, it suffices to show that for the above defined $\rho$ and $f^\beta$ the cost $\fcost(\rho, f^\beta)$ is at most the cost of $f|_{V(G_t)}$ in $G_t$.
Note that the cost of $f|_{V(G_t)}$ in $G_t$ equals the cost of $f|_{\beta(t)}$ in $G_t[\beta(t)]$ plus the sum over all $s \in Z$ of the cost of $f|_{V(G_s)}$ in $G_s$.
Consider the types of constraints one by one.

First, consider a child constraint $\Gamma(s)$ for some $s \in Z$. By the definition of $\rho(\Gamma(s))$, $f$ satisfies the requirements for Point~\ref{i:ac:sound}
for the cell $M[s, \rho(\Gamma(s)), f|_{\adh(s)}]$. Consequently, the cost of $f|_{V(G_s)}$ in $G_s$ is not smaller than 
$M[s, \rho(\Gamma(s)), f|_{\adh(s)}] = M_{\Gamma(s)}(\rho(\Gamma(s)), f|_{\adh(s)})$. 
The latter term is exactly the contribution of the constraint $\Gamma(s)$ to the cost $\fcost(\rho,f^\beta)$.

Second, consider an edge $e = uv \in E(G_t[\beta(t)])$. Note that $\rho(\Gamma(e)) = \emptyset$ by definition while $M_{\Gamma(e)}[ \emptyset, f|_e]$ equals $1$
if $f(u) \neq f(v)$ and $0$ otherwise. Thus, the contribution of the edge $e$ towards the cost of $f|_{\beta(t)}$ in $G_t[\beta(t)]$ is the same as the contribution
of $\Gamma(e)$ towards $\fcost(\rho,f^\beta)$.

Third, consider a terminal constraint $\Gamma(i,j, w)$. If $\rho(\Gamma(i,j, w)) = \emptyset$, then the contribution of this constraint to $\fcost(\rho,f^\beta)$ is $0$.
Otherwise, we have $\rho(\Gamma(i,j, w)) = \{(i,j)\}$ and this can only happen if $f(w) = j$ and $w \in T_i$. 
By definition, this implies $M_{\Gamma(i,j,w)}(\{(i,j)\}, f|_{\{w\}}) = 0$, and again the contribution of this constraint to $\fcost(\rho,f^\beta)$ is $0$.

This concludes the proof that the cost of $f|_{V(G_t)}$ in $G_t$ is not smaller than $\fcost(\rho, f^\beta)$ and concludes the proof of the claim.
\cqed\end{proof}

By Claim~\ref{cl:ac:feas}, it suffices to minimize $\fcost(\rho, f^\beta)$ over responsibility assignments $\rho$ and feasible functions $f^\beta$ such that $\bigcup_{\Gamma \in \cnstr} \rho(\Gamma) = I$.
Assume that this minimum is at most $k$ and fix some minimizing arguments $(\rho^\ast, f^{\beta,\ast})$.
We say that a constraint $\Gamma$ is \emph{touched} if either $\rho^\ast(\Gamma) \neq \emptyset$ or $f^{\beta,\ast}|_{X_\Gamma}$ is not a constant function. We claim the following
\begin{claim}\label{cl:ac:touched}
There are at most $|I|+k$ touched constraints.
\end{claim}
\begin{proof}
By the assumption that the values $\rho^\ast(\Gamma)$ are pairwise disjoint, there are at most $|I|$ constraints
with $\rho^\ast(\Gamma) \neq \emptyset$. 
Fix a constraint $\Gamma$ such that $f^{\beta,\ast}|_{X_\Gamma}$ is not a constant function. To finish the proof of the claim
it suffices to show that $\Gamma$ contributes at least $1$ to the sum in~\eqref{eq:ac:cost}. 

Clearly, $\Gamma$ is not a terminal constraint. If $\Gamma$ is an edge constraint, $\Gamma = \Gamma(e)$
for some $e = uv \in E(G_t[\beta(t)])$, then as $X_{\Gamma(e)} = e$ we have $f^{\beta,\ast}(u) \neq f^{\beta,\ast}(v)$.
Hence, $\Gamma$ contributes $1$ to the sum in~\eqref{eq:ac:cost}. 

Finally, if $\Gamma = \Gamma(s)$ for some $s \in Z$, then $\Gamma$ contributes $M[s, \rho^\ast(\Gamma(s)), f^{\beta,\ast}|_{\adh(s)}]$
to the sum in~\eqref{eq:ac:cost}. By Point~\ref{i:ac:cmpl}, there exists an extension $f_s$ of $f^{\beta,\ast}|_{\adh(s)}$ to $V(G_s)$ of cost at most $M[s,\rho^\ast(\Gamma(s)),f^{\beta,\ast}|_{\adh(s)}]$. By compactness and the fact that $f^{\beta,\ast}|_{\adh(s)}$ is not a constant function, any such extension has a positive cost in $G_s$. This finishes the proof of the claim.
\cqed\end{proof}

Let $A^\ast = (f^{\beta,\ast})^{-1}([p] \setminus \{\jmaj(f^{\beta,\ast})\})$. Since $f^{\beta,\ast}$ is unbreakable-consistent, we have $|A^\ast| \leq 3k$. 
Let $B^\ast$ be the set of all vertices $v \in \beta(t) \setminus A^\ast$ for which there exists a touched constraint
$\Gamma \in \cnstr$ with $v \in X_\Gamma$. 
By Claim~\ref{cl:ac:touched}, we have that $|B^\ast| \leq k(k+|I|)$.

Our application of Lemma~\ref{lem:random2} is encapsulated in the following claim.
\begin{claim}\label{cl:ac:random}
In time $2^{\Oh((k+|I|) \log (k+|I|))} n^{\Oh(1)}$ one can generate a family $\randfamily$
of pairs $(j,g)$ where $j \in [p]$ and $g \colon \beta(t) \to [p]$.
The family is of size $2^{\Oh((k+|I|) \log (k+|I|))} n$ and
there exists $(j,g) \in \randfamily$ such that $j = \jmaj(f^{\beta,\ast})$ and $g$ agrees with $f^{\beta,\ast}$ on
$A^\ast \cup B^\ast$. Furthermore, for every element $(j,g) \in \randfamily$, $g$ extends $f^\adh$.
\end{claim}
\begin{proof}
First, we iterate over all choices of nonnegative integers $j \in [p]$ and $a_1,a_2, \ldots, a_{p}$
such that $\sum_{i=1}^{p} a_i \leq 3k+k(k+|I|)$ and $\sum_{i \in [p] \setminus \{j\}} a_i \leq 3k$. 
Clearly, there are $2^{\Oh(k)}(|I|+k)$ options as $p \leq k+1$.
For a fixed choice of $j$ and $(a_i)_{i=1}^{p}$ we invoke Lemma~\ref{lem:random2}
for $U = \beta(t) \setminus \adh(t)$ and integers $r = p$, $(a_i)_{i=1}^{p}$, obtaining a family
$\randfamily'$. For every $g' \in \randfamily'$, we insert $(j, g' \cup f^\adh)$
  into $\randfamily$.

The bound on the size of the output family $\randfamily$ follows from bound of Lemma~\ref{lem:random2}
and the inequality $\log^k n \leq 2^{\Oh(k \log k)} n$. 
Finally, note that the promised pair $(j,g)$ will be generated for the choice of
$j = \jmaj(f^{\beta,\ast})$ and $a_i = |(f^{\beta,\ast})^{-1}(i) \cap (A^\ast \cup B^\ast)\setminus \adh(t)|$ for every $i \in [p]$.
\cqed\end{proof}

Invoke the algorithm of Claim~\ref{cl:ac:random}, obtaining a family $\randfamily$. We say
that $(j,g) \in \randfamily$ is \emph{lucky} if $f^{\beta,\ast}$ exists, $j = \jmaj(f^{\beta,\ast})$, and $f^{\beta,\ast}$
agrees with $g$ on $A^\ast \cup B^\ast$. 

Consider now an auxiliary graph $H$ with $V(H) = \beta(t)$ and $uv \in E(H)$ if and only if
$u \neq v$ and there exists a constraint $\Gamma \in \cnstr$ with $u,v\in X_\Gamma$.
By the definition of constraints $\Gamma(e)$, we have that $G_t[\beta(t)]$ is a subgraph of
$H$, but in $H$ we also turn all adhesions $\adh(s)$ for $s \in Z$ into cliques.

For a pair $(j,g) \in \randfamily$, we define $\Sigma(j,g) := g^{-1}([p] \setminus \{j\})$.
Observe the following.

\begin{claim}\label{cl:ac:stains}
For a lucky pair $(j,g)$,
every connected component of $H[\Sigma(j,g)]$ is either completely contained
in $(f^{\beta,\ast})^{-1}(j)$ or completely contained in $A^\ast$.
\end{claim}
\begin{proof}
Assume the contrary, and let $uv \in E(H)$ be an edge violating the condition. By symmetry, assume $f^{\beta,\ast}(u) \neq j$.
Then, since $(j,g)$ is lucky, we have that $u \in A^\ast$. Hence $v \notin A^\ast$, that is, $f^{\beta,\ast}(v) = j$. 
By the definition of $H$, there exists a constraint $\Gamma \in \cnstr$ with $u,v \in X_{\Gamma}$ (either $\Gamma = \Gamma(uv)$ or
$\Gamma = \Gamma(s)$ for some $s \in Z$ with $u,v \in \adh(s)$). Consequently, $\Gamma$ is touched, and $v \in B^\ast$.
However, then from the fact that $(j,g)$ is lucky it follows that $g(v) = j$, a contradiction.
\cqed\end{proof}
Claim~\ref{cl:ac:stains} motivates the following approach. We try every pair $(j,g) \in \randfamily$,
and proceed under the assumption that $(j,g)$ is lucky. 
We inspect connected components of $H[\Sigma(j,g)]$ one-by-one and either try to 
set a candidate for function $f^{\beta,\ast}$ to be equal to $g$ or constantly equal $j$ on the component.
Claim~\ref{cl:ac:stains} asserts that one could construct $f^{\beta,\ast}$ in this manner.
The definition of $H$ implies that for every constraint $\Gamma \in \cnstr$, the set $X_\Gamma$ intersects at most one component
of $H[\Sigma(j,g)]$. This gives significant
independence of the decisions, allowing us to execute a multidimensional knapsack-type dynamic
programming, as between different connected components of $H[\Sigma(j,g)]$ we need only to keep intermediate values of the
cost and the union of values of the constructed responsibility assignment. 

We proceed with formal arguments. 
For every $(j,g) \in \randfamily$, we proceed as follows.
Let $C_1,C_2,\ldots,C_\ell$ be the connected components of $H[\Sigma(j,g)]$ and for $1 \leq i \leq \ell$, 
let $\cnstr_i$ be the set of constraints $\Gamma \in \cnstr_i$ with
$X_\Gamma \cap C_i \neq \emptyset$.
By the definition of $H$,
  the sets $\cnstr_i$ are pairwise disjoint.
Let $\cnstr_0 = \cnstr \setminus \bigcup_{i=1}^\ell \cnstr_i$ be the remaining constraints, i.e., the constraints $\Gamma \in \cnstr$ with $X_\Gamma \cap \Sigma(j,g) = \emptyset$.
It will be convenient for us to denote $C_0 = \emptyset$ to be a component accompanying
$\cnstr_0$.
For $J \subseteq \{0,1,2,\ldots,\ell\}$, denote $g_J$ to be the function $g$ modified as follows: 
for every $i \in \{0,1,\ldots,\ell\} \setminus J$ and every $v \in C_i$ we set $g_J(v) = j$.
Furthermore, let $J^\adh \subseteq \{1,2,\ldots,\ell\}$ be the set of these indices $j$
for which $C_j \cap \adh(t) \neq \emptyset$.
For $0 \leq i \leq \ell$, denote $\cnstr_{\leq i} = \bigcup_{j \leq i} \cnstr_j$.

Let $m = |\cnstr|$. We order constraints in $\cnstr$ as $\Gamma_1,\Gamma_2,\ldots,\Gamma_m$
according to which set $\cnstr_i$ they belong to. That is, if $\Gamma_a \in \cnstr_i$,
$\Gamma_b \in \cnstr_j$ and $i < j$, then $a < b$.
For $0 \leq i \leq \ell$,
let $\lasta(i) = |\cnstr_{\leq i}|$, that is, $\Gamma_{\lasta(i)} \in \cnstr_{\leq i}$
but $\Gamma_{\lasta(i)+1}$ does not belong to $\cnstr_{\leq i}$ (if exists).
Furthermore, let $\lasta(-1) = 0$.

Let $0 \leq i \leq \ell$ and let $\lasta(i-1) \leq a \leq \lasta(i)$.
An \emph{$a$-partial responsibility assignment} $\rho$ is a responsibility assignment
defined on constraints $\Gamma_1, \ldots, \Gamma_a$ such that $\bigcup_{b=1}^a \rho(\Gamma_b) \subseteq I$.
For a set $J \subseteq \{0,1,\ldots,i\}$ containing $J^\adh \cap \{0,1,\ldots,i\}$
and an $a$-partial responsibility assignment $\rho$, we define the cost of $J$ and $\rho$ as
$$\fcost_{i,a}(J, \rho) = \sum_{b=1}^a M_{\Gamma_b}(\rho(\Gamma_b), g_J|_{X_{\Gamma_b}}).$$
Furthermore, let $k^\beta = k$ if $|\beta(t)| > 3k$ and $k^\beta = |\beta(t)|$ otherwise.
The goal of our dynamic programming algorithm is to compute, for every $-1 \leq i \leq \ell$, $0 \leq k^\bullet \leq k^\beta$,
 and $I^\bullet \subseteq I$
a value $Q_i[k^\bullet, I^\bullet]$ that equals a minimum possible cost $\fcost_{i,\lasta(i)}(J,\rho)$ over all $J \subseteq \{0,1,\ldots,i\}$ containing $J^\adh \cap \{0,1,\ldots,i\}$ with $|\bigcup_{i' \in J} C_{i'}| \leq k^\bullet$,
and $\lasta(i)$-partial responsibility assignment $\rho$
with $\bigcup_{b=1}^{\lasta(i)} \rho(\Gamma_b) = I^\bullet$.
The next claim shows that it suffices.
\begin{claim}\label{cl:ac:dp}
We have $Q_\ell[k^\beta, I] \geq \fcost(\rho^\ast, f^{\beta,\ast})$. 
Furthermore, if $\fcost(\rho^\ast,f^{\beta,\ast}) \leq k$ and $(j,g)$ is lucky, then 
$Q_\ell[k^\beta,I] = \fcost(\rho^\ast, f^{\beta,\ast})$.
\end{claim}
\begin{proof}
For the first claim, let $(J, \rho)$ be the witnessing arguments for the value $Q_\ell[k^\beta, I]$.
Note that $\rho$ is a responsibility assignment. Furthermore, $g_J$ is a feasible function: it extends $f^\adh$
due to the requirement $J^\adh \subseteq J$ and it is unbreakable-consistent
due to the requirement $|\bigcup_{i' \in J} C_{i'}| \leq k^\beta$. 
The first claim follows from the minimality of $(\rho^\ast, f^{\beta,\ast})$.

For the second claim, by Claim~\ref{cl:ac:stains} we know that on every connected component $C_i$,
the function $f^{\beta,\ast}$ either equals $g$ or is constant at $j$. 
Let $J^\ast$ be the set where the first option happens; note that $J^\adh \subseteq J^\ast$
and $g_{J^\ast} = f^{\beta,\ast}$. Consequently, $\fcost_{\ell,m}(J^\ast, \rho^\ast) = \fcost(\rho^\ast, f^{\beta,\ast})$. 
Finally, as $f^{\beta,\ast}$ is unbreakable-consistent, we have $|\bigcup_{i' \in J^\ast} C_{i'}| \leq k^\beta$.
The claim follows from the minimality of $Q_\ell[k^\beta, I]$.
\cqed\end{proof}

For the initialization step of our dynamic programming algorithm, note that for $i=-1$
we have $\lasta(i)=0$ and the minimization for $Q_{-1}[\cdot,\cdot]$
takes into account only $J = \emptyset$ and $\rho=\emptyset$.

Fix now $0 \leq i \leq \ell$, we are to compute values $Q_i[\cdot,\cdot]$. 
To this end, we use another layer of dynamic programming.
For $\lasta(i-1) \leq a \leq \lasta(i)$, $0 \leq k^\bullet \leq k^\beta$, and $I^\bullet \subseteq I$
we define $Q_{i,a}^\in[k^\bullet, I^\bullet]$ to be a minimum possible cost
$\fcost_{i,a}(\rho,J)$ over all $J \subseteq \{0,1,2,\ldots,i\}$ containing
$J^\adh \cap \{0,1,2,\ldots,i\}$ satisfying that $i\in J$ and
$|\bigcup_{i' \in J} C_{i'}| \leq k^\bullet$, and all $a$-partial responsibility assignments
$\rho$ with $\bigcup_{b=1}^a \rho(\Gamma_b) = I^\bullet$.

Similarly we define $Q_{i,a}^{\notin}[k^\bullet, I^\bullet]$ with the requirement that $i \notin J$.
This definition is applicable only if $i \notin J^\adh$; for indices $i \in J^\adh$ we do not compute the
values $Q_{i,a}^{\notin}$.

Since $Q_{i,\lasta(i-1)}^\in$, $Q_{i,\lasta(i-1)}^{\notin}$, and $Q_{i-1}$ use all constraints up to $\lasta(i-1)$,
we have
$$Q_{i,\lasta(i-1)}^\in[k^\beta + |C_i|, I^\bullet] = Q_{i,\lasta(i-1)}^{\notin}[k^\beta, I^\bullet] = Q_{i-1}[k^\beta, I^\bullet].$$
Here, and in subsequent formulae, we assume that a value of a cell $Q_{i,a}^\in[\cdot,\cdot]$
or $Q_i[\cdot,\cdot]$ equals $0$ if the first argument is negative.
The above formula serves as the initialization step for computing values $Q_{i,a}^\in[\cdot,\cdot]$
and $Q_{i,a}^{\notin}[\cdot,\cdot]$.

For a single computation step, fix $\lasta(i-1) < a \leq \lasta(i)$.
The definitions of $Q_{i,a}^\in[\cdot]$ and $Q_{i,a-1}^\in[\cdot,\cdot]$ differ only in the requirement
to define $f(\Gamma_a)$. Furthermore, $X_{\Gamma_a}$ intersects only the component $C_i$
(if $i > 0$), while $i$ is required to be contained in $J$ in the definition
of both $Q_{i,a}^\in[\cdot,\cdot]$ and $Q_{i,a-1}^\in[\cdot,\cdot]$. Consequently,
$$Q_{i,a}^\in[k^\bullet, I^\bullet] = \min_{I' \subseteq I^\bullet} \left(Q_{i,a-1}^\in[k^\bullet, I^\bullet \setminus I'] + M_{\Gamma_a}(I', g|_{X_{\Gamma_a}})\right).$$
  Similarly, for $i \notin J^\adh$ we have  
$$Q_{i,a}^{\notin}[k^\bullet, I^\bullet] = \min_{I' \subseteq I^\bullet} \left(Q_{i,a-1}^{\notin}[k^\bullet, I^\bullet \setminus I'] + M_{\Gamma_a}(I', X_{\Gamma_a} \times \{j\})\right).$$
  Finally, we set for $i \in J^\adh$:
$$Q_i[k^\bullet, I^\bullet] = Q_{i,\lasta(i)}[k^\bullet, I^\bullet],$$
and for $i \notin J^\adh$:
$$Q_i[k^\bullet, I^\bullet] = \min \left( Q_{i,\lasta(i)}^\in[k^\bullet, I^\bullet], Q_{i,\lasta(i)}^{\notin}[k^\bullet, I^\bullet]\right)$$

The above dynamic programming algorithm computes the values $Q_\ell[\cdot,\cdot]$ in time $3^{|I|} \cdot n^{\Oh(1)}$.
By Claims~\ref{cl:ac:feas} and~\ref{cl:ac:dp}, we can take $M[t, I, f^\adh]$
to be the minimum value of $Q_\ell[k^\beta, I]$ encountered over all choices of $(j,g) \in \randfamily$.
Claim~\ref{cl:ac:M} shows that the above suffices to conclude the proof of Theorem~\ref{thm:auxcut}.
\end{proof}

\subsubsection{Steiner Cut}

\begin{proof}[Proof of Theorem~\ref{thm:steinercut}.]
Let $(G,T,p,k)$ be an input to \textsc{Steiner Cut} and let $n = |V(G)|$. 

First, we observe that it is sufficient to solve \textsc{Steiner Cut} with the additional assumption that $G$ is connected.
Indeed, in the general case we can proceed as follows. First, we discard all connected components of $G$
that are disjoint with $T$. 
Second, for every connected component $G'$ of $G$ and every $0 \leq k' \leq k$, $0 \leq p' \leq p$, we solve \textsc{Steiner Cut}
on the instance $(G', T \cap V(G'), p', k')$.
Finally, the results of these computations allow us to solve the input instance by a straightforward knapsack-type
dynamic programming algorithm.

Thus, we assume that $G$ is connected. In particular, we can assume that $p \leq k+1$ as otherwise the input
instance is a trivial no-instance.
With these assumptions, it is straightforward to observe that the input \textsc{Steiner Cut} instance $(G,T,p,k)$
is equivalent to \textsc{Auxiliary Cut} instance $(G,p,k,(T),\{1\} \times [p])$, that is, we set $\tau = 1$,
   $T_1 = T$, and $I = \{1\} \times [p]$. Theorem~\ref{thm:steinercut} follows.
\end{proof}

\subsubsection{Steiner Multicut}

\begin{proof}[Proof of Theorem~\ref{thm:multicut}.]
Let $(G,(T_i)_{i=1}^t,k)$ be an input to \textsc{Steiner Multicut} and let $n = |V(G)|$.

First, we observe that it is sufficient to solve \textsc{Steiner Multicut} with the additional assumption that $G$ is connected.
Indeed, in the general case we can proceed as follows. First, we discard all sets $T_i$ that intersect more than one connected component of $G$. Second, for every connected component $G'$ of $G$ and every $0 \leq k' \leq k$, we solve \textsc{Steiner Multicut}
on graph $G$, terminal sets $\{T_i~|~T_i \subseteq V(G')\}$, and budget $k'$.
Finally, the results of these computations allow us to solve the input instance by a straightforward knapsack-type
dynamic programming algorithm.
Thus, henceforth we assume that $G$ is connected.
In particular, any function $f : V(G) \to \mathbb{Z}$ of cost at most $k$ can attain at most $k+1$ distinct values. Let $p = k+1$.

If $G$ is connected, then the input \textsc{Steiner Multicut} instance $(G,(T_i)_{i=1}^t,k)$ is a yes-instance
if and only if there exists $f \colon V(G) \to [p]$ of cost at most $k$ such that for every $i \in [t]$, $f$ is not constant on $T_i$.
Thus, to reduce \textsc{Steiner Multicut} to \textsc{Auxiliary Multicut}, it suffices to guess, for every $i \in [t]$, two distinct values from $[p]$ than $f$ attains on $T_i$.

More formally, we iterate over all sequences $(M_i)_{i=1}^t$ such that for every $i \in [t]$ we have $M_i \subseteq [p]$
and $|M_i| = 2$. There are $2^{\Oh(t \log t)}$ such sequences. 
For every sequence $(M_i)_{i=1}^t$ we define $I = \bigcup_{i=1}^t \{i\} \times M_i$ and invoke
the algorithm of Theorem~\ref{thm:auxcut} for \textsc{Auxiliary Cut} instance $(G,p,k,(T_i)_{i=1}^t,I)$. 
By the discussion above, the input \textsc{Steiner Multicut} instance is a yes-instance if and only
if at least one of the generated \textsc{Auxiliary Cut} instances is a yes-instance. 
Since in every call we have $|I| = 2t$, Theorem~\ref{thm:multicut} follows.
\end{proof}

\section{Conclusions}

In this paper we presented an algorithm that constructs a tree decomposition as in~\cite{minbisection-STOC}, but faster, with better parameter bounds, and arguably simpler. 
This allowed us to improve the parametric factor in the running time bounds for a number of parameterized algorithms for cut problems to $2^{\Oh(k \log k)}$.
We conclude with a few open questions.

First, can we construct the decomposition as in this paper in time $2^{\Oh(k)} n^{\Oh(1)}$? The current techniques seem to fall short of this goal.


The problem of finding a lean witness for a bag that is not $(k,k)$-unbreakable can be generalized to the following {\sc{Densest Subhypergraph}} problem:
given a hypergraph $H=(V,E)$, where the same hyperedge may appear multiple times (i.e. it is a multihypergraph), and an integer $k$,
one is asked to find a set $X\subseteq V$ consisting of $k$ vertices that maximizes the number of hyperedges contained in $X$.
A trivial approximation algorithm 
is to just find the hyperedge of size at most $k$ that has the largest multiplicity and cover it; this achieves the approximation factor of $2^k$.
We conjecture that it is not possible to obtain a $2^{o(k)}$-approximation in FPT time for this problem.

Third, we are not aware of any conditional lower bounds for \textsc{Minimum Bisection} that are even close to our $2^{\Oh(k \log k)} n^{\Oh(1)}$ upper bound.
To the best of our knowledge, even an algorithm with running time $2^{o(k)}\cdot n^{\Oh(1)}$ is not excluded under the Exponential Time Hypothesis using the known {\sc{NP}}-hardness reductions.
It would be interesting to understand to what extent the running time of our algorithm is optimal.


\bibliography{lean}

\appendix

\section{Making a tree decomposition compact: Proof of Lemma~\ref{lem:compactification}}\label{app:compactification}
\begin{proof}[Proof of Lemma~\ref{lem:compactification}.]
We assume that $(T,\beta)$ is rooted by rooting it at any vertex.
By performing first the cleanup operation, we ensure that $T$ has at most $n=|V(G)|$ edges.


Having cleaned $(T,\beta)$, we proceed to the construction of $(\wh{T},\wh{\beta})$.
We gradually modify the decomposition $(T,\beta)$ as long as it is not compact, and output it as $(\wh{T},\wh{\beta})$ once the compactness is achieved.
Suppose then that $(T,\beta)$ is still not compact. Then there exists a node $t$ that violates
the definition of compactness; we proceed as follows. 

First, observe that the properties of a tree decomposition imply that $N_G(\alpha(t)) \subseteq \adh(t)$. 
If there exists $v \in \adh(t) \setminus N_G(\alpha(t))$, then we can delete $v$ from the bag of $t$ and all its descendants.
This operation strictly decreases the sum of sizes of all bags, while every bag and adhesion can only be replaced by its subset.
Hence, we may apply it exhaustively, to all nodes $t$ violating compactness in this manner,
thus arriving in polynomial time at the situation where we can assume that $N_G(\alpha(t)) = \adh(t)$ for every $t\in V(T)$.

Therefore, if now a node $t$ violates the compactness, then we have that $G[\alpha(t)]$ is disconnected.
Let $(A,B)$ be an edge cut of $G[\alpha(t)]$ of zero order with $A,B \neq \emptyset$.
Since $\adh(t) \neq \emptyset$, $t$ is not the root of $T$; let $s$ be the parent of $t$.
Let $T_t$ be the subtree of $T$ rooted in $t$. We make two copies of $T_t$, $T_t^A$ and $T_t^B$,
and define the tree $T'$ as $T$ with $T_t$ replaced with $T_t^A$ and $T_t^B$, both with roots
being children of $s$. Furthermore, we define $\beta'\colon V(T') \to 2^{V(G)}$ as:
$\beta'(p) = \beta(p) \cap (A \cup \adh(t))$ for $p \in V(T_t^A)$,
$\beta'(p) = \beta(p) \cap (B \cup \adh(t))$ for $p \in V(T_t^B)$ and
$\beta'(p) = \beta(p)$ otherwise.
It is straightforward to verify that $(T',\beta')$ is a tree decomposition of $G$.
Moreover, every bag of $(T',\beta')$ is a subset of a bag of $(T,\beta)$, and similarly for adhesions.
Finally, we clean up $(T',\beta')$ so that it has at most $n$ edges using the same operation as in the first paragraph.
We conclude by replacing $(T,\beta)$ with $(T',\beta')$ and applying the reasoning again.

To bound the number of iterations of the presented procedure, consider the potential $\sum_{t \in V(T)} |\alpha(t)|^2$.
Note that its value is integral and initially bounded by $\Oh(n^3)$.
It is straightforward to verify that the potential strictly decreases in every iteration, so the procedure stops after at most $\Oh(n^3)$ iterations.
Since every iteration can be performed in polynomial time, the whole algorithm works in polynomial time as well.
\end{proof}

\end{document}